\documentclass[11pt]{article}
\usepackage[left=0.5in,right=0.5in,bottom=0.75in,top=0.75in]{geometry}

\usepackage{amsmath}
\usepackage{amssymb}
\usepackage{amsthm}
\usepackage{graphicx}
\graphicspath{{code/out/}}
\usepackage{natbib}
\usepackage{hyperref}
\hypersetup{colorlinks=true, citecolor=black, linkcolor=., urlcolor=cyan}
\usepackage[font = small]{caption}
\usepackage[shortlabels]{enumitem}
\usepackage{booktabs}
\usepackage{algorithm}
\usepackage{algpseudocode}
\algdef{SxnE}[FOR]{For}{EndFor}[1]{\algorithmicfor\ #1 }
\usepackage{tikz}
\usetikzlibrary{tikzmark}
\usepackage{array} 
\usepackage[normalem]{ulem}
\usepackage{appendix}
\usepackage{url}

\newtheorem{thm}{Theorem}
\algnewcommand\algorithmicswitch{\textbf{switch}}
\algnewcommand\algorithmiccase{\textbf{case}}
\algnewcommand\algorithmicassert{\textbf{then }}
\algnewcommand\Assert[1]{\State \algorithmicassert #1}%
\algnewcommand\Assertif[2]{\State \textbf{if} #1 \algorithmicassert #2}
\algnewcommand\Assertelse[1]{\State \textbf{else} #1}
\algdef{SE}[SWITCH]{Switch}{EndSwitch}[1]{\algorithmicswitch\ #1}{\algorithmicend\ \algorithmicswitch}%
\algdef{SE}[CASE]{Case}{EndCase}[1]{\algorithmiccase\ #1}{\algorithmicend\ \algorithmiccase}%
\algtext*{EndSwitch}%
\algtext*{EndCase}%

\providecommand{\I}{\mathbf{I}}

\providecommand{\Q}{\mathbf{Q}}

\providecommand{\X}{\mathbf{X}}

\newcommand{\x}{\mathbf{x}}
\newcommand{\y}{\mathbf{y}}

\let\origc\c
\DeclareRobustCommand\c{\ifmmode\mathbf{c}\else\expandafter\origc\fi}
\let\origd\d
\DeclareRobustCommand\d{\ifmmode\mathbf{d}\else\expandafter\origd\fi}
\let\origu\u
\DeclareRobustCommand\u{\ifmmode\mathbf{u}\else\expandafter\origu\fi}
\let\origd\v
\DeclareRobustCommand\v{\ifmmode\mathbf{v}\else\expandafter\origv\fi}

\providecommand{\cA}{\mathcal{A}}

\providecommand{\bb}{\boldsymbol{\beta}}
\providecommand{\bh}{\widehat{\beta}}
\providecommand{\bbh}{\widehat{\boldsymbol{\beta}}}

\providecommand{\veps}{\varepsilon}

\providecommand{\bvep}{\boldsymbol{\varepsilon}}

\providecommand{\lam}{\lambda}

\providecommand{\bt}{\boldsymbol{\theta}}

\renewcommand{\Pr}{\mathbb{P}}

\providecommand{\cor}{\textrm{Cor}}

\providecommand{\Norm}{\textrm{N}}

\providecommand{\Tr}{^{\scriptscriptstyle\top}}

\providecommand{\CV}{\textrm{CV}}

\usepackage[scr=boondox]{mathalpha}

\makeatletter
\@ifpackagelater{mathalpha}{2021/01/01}{%

}{%

}
\makeatother

\providecommand{\iid}{\overset{\text{iid}}{\sim}}

\providecommand{\abs}[1]{\left\lvert#1\right\rvert}

\providecommand{\norm}[1]{\lVert#1\rVert}

\providecommand{\gg}{\succ}

\providecommand{\al}[2]{\begin{align}\label{#1}#2\end{align}}
\providecommand{\as}[1]{\begin{align*}#1\end{align*}}

\providecommand{\cvr}{\textrm{Cover}}

\newlength{\li} 
\newcommand{\singlespace}{\baselineskip 1.1\li}

\renewcommand{\abstract}[1]{
 \centerline{
 \begin{minipage}{0.7\linewidth}
 \hrule
 \vskip 0.1in
  \begin{center}
    {\bf Abstract}
  \end{center}
  #1
 \vskip 0.1in
 \hrule
 \end{minipage}}
 \vskip 0.3in}

\title{A New Perspective on High Dimensional Confidence Intervals}
\author{
  Logan Harris\\Department of Biostatistics\\University of Iowa
  \and
  Patrick Breheny\\Department of Biostatistics\\University of Iowa
}
\date{\today}

\begin{document}

\maketitle

\abstract{
Classically, confidence intervals are required to have consistent coverage across all values of the parameter. However, this will inevitably break down if the underlying estimation procedure is biased. For this reason, many efforts have focused on debiased versions of the lasso for interval construction. In the process of debiasing, however, the connection to the original estimates are often obscured. In this work, we offer a different perspective focused on average coverage in contrast to individual coverage. This perspective results in confidence intervals that better reflect the original assumptions, as opposed to debiased intervals, which often do not even contain the original lasso estimates. To this end we propose a method based on the Relaxed Lasso that gives approximately correct average coverage and compare this to debiased methods which attempt to produce correct individual coverage. With this new definition of coverage we also briefly revisit the bootstrap, which Chatterjee and Lahiri (2010) showed was inconsistent for lasso, but find that it fails even under this alternative coverage definition.
}

\section{Introduction}

The objective function for lasso-penalized linear regression \citep{Tibshirani1996} is
$$Q(\bb|\X,\y,\lambda) = \frac{1}{2n}\norm{\y - \X\bb}_2^2 + \lambda\norm{\bb}_1,$$
where $\y$ is a length $n$ vector of independent outcomes, $\X$ is an $n \times p$ matrix of features, $\bb$ is a length $p$ vector of regression coefficients, and $\lambda$ is a regularization parameter controlling the amount of penalization. Note that the objective function involves the addition of the $L_1$ penalty, $\lambda\norm{\bb}_1 = \lambda \sum_{j = 1}^p |\beta_j|$, to the squared error loss. This typically results in sparse estimates for some of the regression coefficients (i.e., $\bh_j = 0$) depending on the choice of the regularization parameter $\lambda$. Its ability to carry out both variable selection and estimation is particularly attractive, especially in scenarios where both predictive accuracy and interpretability are important. The lasso performs particularly well in cases where the number of features is large and the underlying model is sparse \citep{HTF2009}, but has become popular in a wide variety of settings.

Nevertheless, inference for the lasso has proven challenging. By introducing both sparsity and shrinkage, the $L_1$ penalty greatly complicates the sampling distribution of the estimators. This complexity has given rise to a wide variety of inferential approaches. The majority of these approaches have focused on controlling the false discovery rate (FDR) of the selected features. Examples include the Covariance test \citep{Lockhart2014}, the Knockoff Filter \citep{Candes2015,Candes2018}, the marginal FDR \citep{Breheny2019}, and the Gaussian mirror \citep{Xing2023}.

There have also been various proposals for constructing confidence intervals (CIs), although the shrinkage/bias introduced by the $L_1$ penalty poses a number of challenges here. Several methods \citep{ZhangZhang2014, Javanmard2014} focus on ``debiasing'' the original point estimates from a lasso fit to facilitate more traditional forms of inference. An alternative approach, which accounts for the uncertainty in model selection by conditioning on the selected model, is known as Selective Inference \citep{Lee2016}, although it is worth noting that this approach only produces intervals for variables that were selected.

In this manuscript, we offer a different perspective that allows for biased intervals and focuses on correct \emph{average} coverage instead of correct \emph{individual} coverage. This perspective results in confidence intervals that better reflect the original assumptions that motivated the use of the lasso for estimation --- as opposed to debiased intervals, which often do not even contain the original lasso estimates. The goal of this paper is not to argue that either definition is inherently superior, but rather to explore the differences between them; we hope the reader finds the debate illuminating.

Section 2 examines the underlying concept of average coverage in more detail and shows that the bootstrap does not even produce average coverage. Section 3 introduces a method based on the Relaxed Lasso which does have approximately correct average coverage. Then Section 4 examines the performance of the proposed method across a number of simulations and includes a comparison to Selective Inference and the de-sparsified lasso. Lastly, in Section 5, we show the application of the proposed method to two data sets, one for acute respiratory illness and the other for gene expression data in mammalian eyes. For the sake of simplicity, we focus on lasso-penalized linear regression, but most of the discussion is relevant to all penalized regression models.

\section{Average coverage}
\label{Sec:difficulties}


When using penalized regression, we are introducing bias into the estimators by design. This has direct implications for confidence intervals constructed around these biased estimates. All methods we are aware of propose debiasing as a way to counteract (either directly or indirectly) the bias introduced in attempts to obtain traditional frequentist coverage properties. In Section~\ref{Sec:IAC}, we instead propose an alternate perspective that focuses on targeting average coverage inspired by the connection between the penalties in penalized regression and Bayesian priors. Given the connection between the bootstrap and Bayesian posteriors, one might suppose that bootstrap confidence intervals also meet this alternate definition for coverage. However, we show in Section~\ref{Sec:boot-bias} that bootstrap intervals fall increasingly short of average coverage as the dimension grows.

\subsection{Individual vs average coverage}
\label{Sec:IAC}

Classical frequentist inference is concerned with achieving proper coverage for each parameter individually. In penalized regression, bias is explicitly being introduced in the estimation procedure which poses a problem when targeting nominal interval coverage for individual parameters. Here we propose shifting focus to average coverage across all $p$ confidence intervals. In penalized regression, these two definitions of coverage can be quite different.

Letting $\cA(\y)$ denote a process that produces an interval based on data $\y$, the coverage probability for the process is defined as $\cvr(\theta) = \Pr\{\theta \in \cA(\y)\}$. Classical frequentist inference requires valid intervals to satisfy $\cvr(\theta) = 1 - \alpha$ for all values of $\theta$ (or potentially $\ge 1 - \alpha$). This is, however, incompatible with Bayesian inference. Bayesian credible intervals cannot, in general, have the same coverage for each $\theta$. What they satisfy instead is maintaining the expected coverage with respect to the prior distribution of $\theta$: $\int \cvr(\theta)p(\theta) \, d\theta = 1 - \alpha$ (this is not the definition of credibility, but it is a consequence, as we show later in this section). Unless the prior is uniform, the coverage of a Bayesian credible interval will be greater than $1-\alpha$ for some $\theta$ and less than $1-\alpha$ for other values of $\theta$.

For example, consider the credible intervals for $\theta$ in a $\Norm(\theta, \sigma^2)$ model with prior $\theta \sim \Norm(0, \tau^2)$ with $\sigma = \tau = 1$. The left side of Figure~\ref{Fig:laplace} illustrates the coverage probability for the 80\% Bayesian credible interval over a range of $\theta$ values. Where the prior density for $\theta$ is highest, the coverage is above 80\%, whereas regions where the prior density is low have coverage below 80\%. The expected coverage, however, is exactly 80\% when integrated with respect to the prior. This is fundamentally true of any Bayesian model with an non-uniform prior: $\cvr(\theta) = 1 - \alpha$ for all values of $\theta$ will never be satisfied. The right side, which illustrates that a similar phenomenon happens for the method we propose, will be discussed in Section~\ref{Sec:coverage}.

\begin{figure}[htb!]
  \begin{center}
    \includegraphics[width=0.8\linewidth]{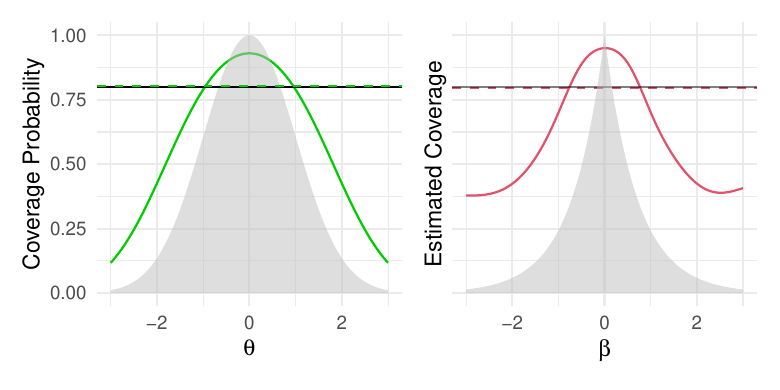}
    \caption{\label{Fig:laplace}
      Coverage probabilities of ridge and lasso (RL-P) confidence intervals. Solid black lines indicates nominal coverage and dashed lines represent average coverage. Left: Exact coverage probabilities across a range of $\theta$ values (see Section~\ref{Sec:difficulties}). Shaded background indicates the normal prior. Right: Empirical coverage probabilities from the simulation in Section~\ref{Sec:coverage} (the curved line is a smooth fit using a binomial GAM). Shaded background shows the Laplace distribution.}
  \end{center}
\end{figure}

In high dimensional problems, there is yet another quantity we can consider: the average coverage. Rather than integrating over a hypothetical distribution of $\theta$ values, we can average over the empirical distribution of parameter values. In other words, we might choose to require that our intervals satisfy $\tfrac{1}{p} \sum_{j=1}^p \cvr(\theta_j) = 1-\alpha$. This criteria is more closely aligned with the Bayesian perspective than a classical frequentist perspective, although it does not specifically require or involve a prior.

Our goal in this paper is not to argue that one of these perspectives is correct and the other is wrong, but rather that the average coverage perspective is reasonable and worthy of consideration. It should not be taken for granted that classical ideas developed for single parameter inference are the best way to approach simultaneous inference for large numbers of parameters. Furthermore, the Bayesian perspective seems to make sense in the context of penalized regression, since penalized regression is intentionally imposing shrinkage towards a prior notion of which parameter values are more likely.

In Section~\ref{Sec:methods}, we propose a new method and in Section~\ref{Sec:results} we see that the resulting intervals, while they do not satisfy classical coverage requirements, perform quite well with respect to the average coverage criterion. We end this section with a short theorem making the explicit connection between Bayesian credible intervals and average coverage.

\begin{thm}
  \label{Thm:bcc}
  If the likelihood is correctly specified according to the true data generating mechanism $p(\y | \bt)$, then a $1-\alpha$ credible set for any parameter $\theta_j$ will satisfy $\int \cvr(\theta_j)p(\bt) d\bt = 1 - \alpha$.
\end{thm}

\begin{proof}
  By definition, a $100(1-\alpha)\%$ credible region for $\theta_j$ is any set $\cA_j(\y)$ such that $\int I\{\theta_j \in \cA_j(\y)\} p(\bt|\y)\,d\bt = 1 - \alpha$. The coverage probability, meanwhile, is defined as $\int I\{\theta_j \in \cA_j(\y)\} p(\y | \bt)d\y$. The average coverage, integrated with respect to the prior $p(\bt)$, is therefore

  \as{
  \int \int I\{\theta_j \in \cA_j(\y)\} p(\y | \bt) p(\bt) d\y d\bt
  &=  \int \int I\{\theta_j \in \cA_j(\y)\} p(\bt | \y) p(\y) d\y d\bt \\
  &=  \int \int I\{\theta_j \in \cA_j(\y)\} p(\bt | \y) d\bt p(\y) d\y \\
  &=  \int  (1 - \alpha) p(\y) d\y \\
  &=  1 - \alpha,
  }

  \noindent and as a result the average coverage is equal to the nominal coverage.
\end{proof}

The above theorem concerns a single parameter of interest $\theta_j$ and its credible interval. An immediate corollary of this result is that the average coverage of $p$ such intervals will also be equal to $1-\alpha$.

A more interesting question is what happens when we average with respect to the empirical distribution of $\theta$ values present in a high-dimensional problem rather than integrating with respect to the prior. In other words, is it true that
\as{\tfrac{1}{p} \sum_{j=1}^p \cvr(\theta_j) \approx \int \cvr(\theta)p(\theta) \, d\theta?}
Note that the left-hand side does not require a Bayesian perspective as no probability distributions of parameters are involved, only the empirical distribution of different values for different parameters. The interpretation here is simpler than that of a conventional frequentist confidence interval: rather than appealing to hypothetical intervals for hypothetical alternative data sets, we are making a more concrete statement here about the $p$ intervals that have just been constructed.

Intuitively, it would seem that the equation above should be true if this empirical distribution of $\theta$ values resembles the prior implied by the penalty and as long as the intervals constructed arise from a distribution resembling a Bayesian posterior. Given the connection between the Bayesian posterior and the bootstrap first pointed out by \cite{Rubin1981}, this would seem to suggest that bootstrap based intervals should satisfy the above equation, but, in Section~\ref{Sec:boot-bias}, we explain why this is not the case. However, in Section~\ref{Sec:methods} we propose an alternative CI construction method inspired by the Bayesian posterior and in Section~\ref{Sec:robustness}, we find that this relationship generally holds for this approach even if the distribution of $\theta$ values is quite different from the prior implied by the lasso penalty.

\subsection{Does the bootstrap give average coverage?}
\label{Sec:boot-bias}

The connection between the bootstrap and a Bayesian posterior was first drawn by \cite{Rubin1981} and further explored by \cite{efron1982} and \cite{Lo1987}. However, these works focused on low dimensional and asymptotic settings where $n \gg p$. \cite{Chatterjee2010} demonstrated that when applied to lasso estimators, the bootstrap is inconsistent -- even if the lasso itself is $\sqrt{n}$-consistent with respect to estimating $\bb$, showing that the frequentist properties of the bootstrap break down in high dimensions. That said, the connection between the bootstrap and a Bayesian posterior along with Theorem~\ref{Thm:bcc} would suggest that perhaps the bootstrap would give correct average coverage. However, here we provide an example showing that the connection between the bootstrap and the Bayesian posterior also breaks down for penalized regression, a problem that becomes much more noticeable when the number of parameters increases. As a consequence, the average coverage of the bootstrap is well below nominal as the bootstrap introduces ``extra bias''.


To provide an example, we return to ridge regression for two reasons. First, we do not face the complication of having estimates shrunk all the way to zero and second, posterior credible intervals can be computed in closed form for Ridge. The simulation is set up to isolate the effect of increasing dimensionality by increasing $p$ $(20, 100, 200)$ but holding $n = 200$ and $\lambda = 0.4$. For the empirical distribution of $\bb$ to be equivalent to the prior (the ideal scenario as indicated by Theorem~\ref{Thm:bcc}) the prior variance ($\tau^2$) must be set to $\sigma^2$ / $n\lam$, since $\lambda$ is the ratio of the prior precision $(1/\tau^2)$ to the information $(n / \sigma^2)$. In this simulation, $\sigma^2 = 100$, so $\tau^2 = 1.25$ and $\beta_j$ was set to the $j/(p+1)$ quantile of a $\Norm(0,\tau^2=1.25)$ distribution. The elements of $\X$ were generated independently from a $\Norm(0, 1)$ and then $\y$ was generated as $\y = \X\bb + \bvep$, where $\veps_i \iid \Norm(0, \sigma^2)$. For each $p$, 1000 data sets were generated and intervals were constructed using both a pairs bootstrap and a Bayesian posterior. Results are provided in \ref{Fig:ridge_converge}, the dashed lines give the average coverages and the solid lines are the estimated coverages as functions of $\beta$.

Even when dimensionality is low, the bootstrap does not entirely agree with the posterior, and the departure increases as we move from $p=20$ to $p=200$. We find that this issue is even worse for the lasso, potentially due to the additional issue of repeatedly drawing exact zeros when bootstrapping; see \ref{sec:boot-fail}. For further explanations for the breakdown, we refer the reader to \ref{Sup:proof}. \ref{Sup:proof} starts with a simple proof in the 1 and 2 predictor setting for both ridge and lasso showing that while this issue increases with dimensionality, that it is present even in low dimensions. Additionally, these proofs indicate that this bias is heavily dependent on the size of the penalty ($\lam$). This is followed by a simulation that decomposes the source of the bootstrap bias in a high dimensional setting for lasso.

\section{Relaxed Lasso Posterior confidence intervals}\label{Sec:methods}


While the bootstrap is not a viable option as outlined in Section~\ref{Sec:boot-bias}, this section and the discourse around Theorem~\ref{Thm:bcc} suggest that if intervals are constructed from a distribution resembling a Bayesian posterior that they should have correct average coverage. We propose the \textbf{Relaxed Lasso Posterior} (RL-P), which constructs intervals from the distribution of $\beta_j$ conditional on the selected features, viewed as a Bayesian posterior. The remainder of this section presents its specific application to lasso-penalized linear regression. Specifically, we define and derive the conditional distributions needed for the interval construction. In this section, we provide a high level derivation of the conditional posterior distributions for lasso-penalized regression. Complete details, including how to calculate quantiles, are provided in \ref{Sup:A}.

As with other penalized regression approaches, the lasso can be formulated as a Bayesian regression model by setting an appropriate prior. This was initially noted by \cite{Tibshirani1996} and explored more extensively by \cite{Park2008}.  For Ridge regression, the prior is a Normal distribution which leads to conjugacy allowing for straightforward interval construction. As seen in both the left side of Figure~\ref{Fig:laplace} and the Ridge Posterior results in \ref{Fig:ridge_converge}, these intervals achieve correct average coverage in ideal settings. Here, for lasso, we derive the conditional distribution of $\bh_j(\lam)$ in attempts to provide intervals analogous to those produced by Ridge.

For the lasso, the corresponding prior is a Laplace distribution, also referred to as the double-exponential distribution:
\as{p(\bb) = \prod_{j = 1}^{p} \frac{\gamma}{2}\exp(-\gamma \abs{\beta_j}), \gamma > 0.}

Let $\hat{S} = \lbrace k: \hat{\beta}_k \neq  0 \rbrace$ denote the set of selected features. Then, let $\hat{S}_j$ denote the set of selected features that excludes feature $j$, so that $\hat{S}_j = \hat{S} \text{ if } j \notin \hat{S}$ and $\hat{S}_j = \hat{S} - \lbrace j \rbrace \text{ if } j \in \hat{S}$. Define $\Q_{\hat{S}_j}$ as $\I - \X_{\hat{S}_j}(\X_{\hat{S}_j} \Tr \X_{\hat{S}_j})^{-1} \X_{\hat{S}_j} \Tr$, the projection matrix onto the features selected by the lasso. The likelihood for $\beta_j$ conditional on the selected features is:

\as{L(\beta_j|\hat{S}_j) \propto \exp(-\frac{\x_j \Tr \Q_{\hat{S}_j} \x_j}{2\sigma^2}(\beta_{j} - \tilde{\beta}_{j})^2)}

\noindent where $\tilde{\beta}_j = (\x_j \Tr \Q_{\hat{S}_j} \x_j)^{-1} \x_j \Tr \Q_{\hat{S}_j} \y$.  This can be seen as a mild extension of the relaxed lasso. It is equivalent to the relaxed lasso for features in $\hat{S}$ but also is capable of providing intervals for features in $\hat{S}^C$.

A normal likelihood and Laplace prior are not conjugate. However, the distribution of $\beta_j$ conditional on $\hat{S}_j $ can be shown to be a composition of right and left truncated normals where the truncation occurs at zero for right and left tails respectively. In this manuscript, we assume that $\X$ has been standardized s.t. $\x_j \Tr\x_j = n$. Then for $\beta_j$ (see \ref{Sup:A} for details),

\al{eq:fcp}{
p(\beta_j | \hat{S}_j) &\propto
\begin{cases}
  C_{-} \exp\{-\frac{\tilde{n}}{2\sigma^2} (\beta_j - (\tilde{\beta}_j + \lambda))^2\}, \text{ if } \beta_j < 0, \\
  C_{+} \exp\{-\frac{\tilde{n}}{2\sigma^2} (\beta_j - (\tilde{\beta}_j - \lambda))^2\}, \text{ if } \beta_j \geq 0 \\
\end{cases}
}
where $\tilde{n} = \x_j \Tr \Q_{\hat{S}_j} \x_j$, $C_{-} = \exp(\tilde{\beta}_j \lambda \tilde{n}/\sigma^2)$ and $C_{+} = \exp(-\tilde{\beta}_j \lambda \tilde{n}/\sigma^2)$. $\tilde{n}$ can be interpreted as the effective sample size.

This formulation is attractive because it allows efficient computation of quantiles. This consists of first determining which normal distribution (left or right tail) the probability corresponds to, then calculating the quantile from the corresponding normal distribution. Again, full details are provided in \ref{Sup:A}.

This solution corresponds to a particular value of $\lam$ and $\hat{\sigma}^2$. Throughout, we select the value of $\lam$ that minimizes cross validation error (CVE), and we estimate $\sigma^2$ as recommended by \citep{Reid2016}:

$$
\hat{\sigma}^2 = \frac{1}{n - |\hat{S}_{\CV}|} ||\y - \X\bbh({\lambda_{\CV}})||_2^2,
$$

\noindent where $|\hat{S}_{\CV}|$ = $\sum \left( \bbh({\lambda_{\CV}}) \neq 0 \right)$.

The relaxed lasso posterior intervals are available through the confidence\_intervals(fit) in the current version of the R package \texttt{ncvreg} (3.16.0).

\section{Results}
\label{Sec:results}

We begin by examining the coverage of the Relaxed Lasso Posterior intervals in what might be considered the ``ideal'' scenario, where the values of $\theta$ match the prior distribution implied by the lasso as discussed in Section~\ref{Sec:difficulties}. We then examine the robustness of the proposed method as the data generating mechanism departs from this ``ideal'' scenario in various ways (Section~\ref{Sec:robustness}). Finally, we compare the proposed confidence interval method to other confidence interval approaches for penalized regression that have been proposed in the literature (Sections~\ref{Sec:Ridge}~and~\ref{Sec:Comparison}), which illustrates the contrast between methods that attempt to debias the intervals and those that do not.

Unless otherwise noted, the nominal coverage rate in all of these experiments is 80\%.

\subsection{Coverage}\label{Sec:coverage}

Given the connection between average coverage and Bayesian credible intervals made by Theorem~\ref{Thm:bcc} and the surrounding discussion, this would suggest that the RL-P method should have approximately correct average coverage when the empirical distribution of $\bt$ matches the prior implied by the lasso penalty, a Laplace (double exponential) distribution.

We generated 1000 independent data sets; for each data set, RL-P intervals were constructed as described in Section~\ref{Sec:methods}. Each data set was simulated as follows. The elements of $\X$ were generated independently from a $\Norm(0, 1)$ with $n = 100$, $p = 101$, and $\beta_j$ was set to the $j/102$ quantile of a Laplace distribution. The coefficients were then scaled so that $\bb \Tr\bb = \sigma^2$, with independent features this results in a signal-to-noise ratio (SNR) of 1. Finally, $\y$ was generated as $\y = \X\bb + \bvep$, where $\veps_i \iid N(0, \sigma^2)$. The results are shown in the right-hand side of Figure~\ref{Fig:laplace}, where the dotted line represents the average coverage across all coefficients, while the solid line is the smoothed estimate of coverage as a function of $\beta$. The black line indicates the nominal coverage rate, which is set to be 80\%.

The RL-P method has average coverage nearly exactly equal to the nominal 80\%. RL-P intervals have high coverage rates for values of $\beta$ near zero and lower coverage rates for values of $\beta$ larger in magnitude. This occurs because for values near zero, the lasso penalty shrinks estimates towards the truth. This leads to a coverage pattern similar to that of Bayesian credible intervals as depicted on the left hand side of Figure~\ref{Fig:laplace} and as described in Section~\ref{Sec:difficulties}.

Low coverage for large values of $\beta$ arises from the bias introduced by the lasso penalty. However, as $n$ increases, the value of $\lambda_{\CV}$ decreases, reducing the bias. Consequently, coverage becomes flatter and closer to the nominal level across the range of $\beta$, as shown in \ref{Sup:alt_ns}.

Bias can also be reduced by choosing an alternative penalty. For example, the method described in Section~\ref{Sec:methods} can be extended to the Minimax Concave Penalty (MCP) to construct Relaxed MCP Posterior (RM-P) intervals. These intervals exhibit coverage that remains much closer to the nominal level across the range of $\beta$ compared to the RL-P intervals (\ref{Sup:MCP}). This demonstrates that the phenomenon of uneven coverage across $\beta$ is not inherent to the proposed confidence interval method, but rather, reflects properties of the penalty used in estimation.

\subsection{Robustness for Average Coverage}
\label{Sec:robustness}

We will now shift our attention to the robustness of the RL-P method under alternative scenarios. We begin by assessing coverage when there is correlation among the predictors. Next, we consider how RL-P performs under various distributions of $\beta$. Finally, we look at how the coverage changes across the range of $\lambda$ values.

\subsubsection{Correlation}
\label{Sec:correlation}

Figure~\ref{Fig:correlation_structure} illustrates the coverage of the RL-P intervals as the level of correlation $\rho = \cor(\x_i, \x_j)$ for $\abs{i-j} = 1$ increases. Otherwise, the simulation design is the same as in Section~\ref{Sec:results}; in fact, the design is exactly the same for $\rho = 0$. The violin plots provide the distributions of average coverages across 1000 simulated data sets for four values of $n$ and three values of $\rho$. For each $n$, the amount of correlation is increased from $\rho = 0$ to $0.5$ to $0.8$.

When $\rho = 0$, RL-P is slightly conservative for $n = 50$ and $n = 100$, but average coverage converges to nominal as $n$ increases. Coverage exceeds the nominal value at smaller $n$ due to the fact that $\sigma^2$ tends to be overestimated at these sample sizes; this phenomenon diminishes as $n$ increases.

Across all sample sizes, RL-P intervals become increasingly conservative as the correlation $\rho$ increases because this leads to more features being selected. The likelihood that a feature with a small effect is selected increases as it becomes more correlated with features that have large effects. This in turn reduces the effective sample size $\tilde{n}$ in Equation~\ref{eq:fcp}, leading to wider intervals. Although this conservative behavior diminishes with larger $n$, it does not disappear entirely.

\begin{figure}[htb!]
  \begin{center}
    \includegraphics[width=0.8\linewidth]{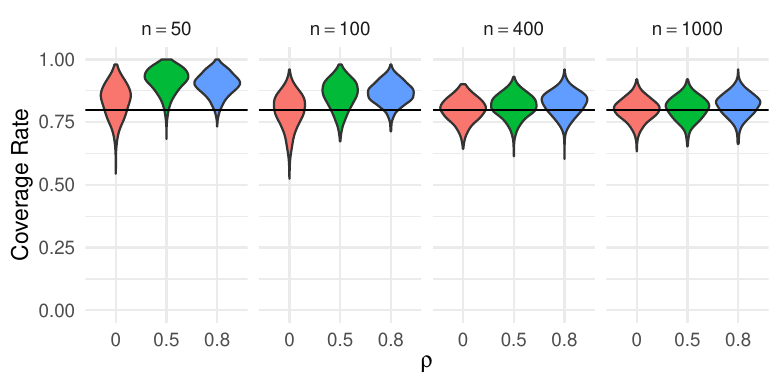}
    \caption{\label{Fig:correlation_structure} This figure presents results for the simulation described in Section~\ref{Sec:correlation}. The violin plots are the distribution of average coverages across 1000 simulated datasets for the RL-P method and across three different levels of autoregressive correlation among the covariates, $\rho = 0 \text{ (no correlation)}, 0.5, 0.8$. For this simulation, p = 100, and the results for each level of correlation are presented for four different sample sizes, $n = p/2, p, 4p, 10p$. The horizontal black line provides reference for the 80\% nominal coverage rate.}
  \end{center}
\end{figure}

\subsubsection{Distribution of Beta} \label{Sec:distribution}

Given the results in Section~\ref{Sec:coverage} that the coverage depends on the magnitude of $\beta$, one might expect that the average coverage is sensitive to the distribution of $\bb$. Table~\ref{Tab:dist_beta} shows the results of $\bb$ distributed as a Laplace as well as 7 alternative distributions. Otherwise, the setup is the same as described in Section~\ref{Sec:results}. Results are shown for 4 sample sizes, $n$ = 50, 100, 400, and 1000. As before, to maintain the specified SNR of 1, $\bb$ is normalized. Prior to normalization, Sparse 1 had $\bb_{1-10} = \pm(0.5, 0.5, 0.5, 1, 2)$ with the rest equal to zero, Sparse 2 had $\bb_{1-31}$ set to 31 evenly distributed quantiles from $N(0, 1)$ with the rest equal to zero, and Sparse 3 had $\bb_{1-51}$ set to 51 evenly distributed quantiles from $N(0, 1)$ with the rest equal to zero. For the T distribution, df was set to 3 and the Beta distribution quantiles were computed from Beta(0.1, 0.1) - 0.5, prior to normalization. The first column of the table provides a visual depiction of these distributions.

The results shown in Table~\ref{Tab:dist_beta} align with what one might expect from Theorem~\ref{Thm:bcc}. First, note that under $\bb$ generated from a Laplace, the average coverage of the RL-P method is slightly conservative for $n=50$ but converges to the nominal rate for the other sample sizes. Distributions that are similar to the Laplace follow this same pattern. For example, when $\bb$ is generated from a T distribution, the coverage rates are nearly identical to the Laplace. When the density / mass is more concentrated near zero, such as with Sparse 1 and 2, the average coverage is above the nominal level. When there is more density away from zero, such as with the normal, the coverage is somewhat below nominal. The worst average coverage occurs when $\beta$ is generated from a Beta(0.1, 0.1) - 0.5 distribution; this is not surprising since the lasso is a poor choice of penalty in this scenario. Even so, the coverage only drops to 69\%, and still converges to the nominal rate as $n$ increases.

\begin{table}[htb!]
  \centering

\begin{tabular}[t]{>{}cccccc}
\toprule
\multicolumn{2}{c}{  } & \multicolumn{4}{c}{Sample Size} \\
\cmidrule(l{3pt}r{3pt}){3-6}
  & Distribution & 50 & 100 & 400 & 1000\\
\midrule
\includegraphics[width=0.67in, height=0.17in]{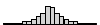} & Laplace & 83.6\% & 79.6\% & 80.0\% & 80.0\%\\
\includegraphics[width=0.67in, height=0.17in]{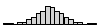} & T & 83.3\% & 78.9\% & 79.8\% & 80.0\%\\
\includegraphics[width=0.67in, height=0.17in]{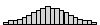} & Normal & 81.6\% & 76.2\% & 79.7\% & 79.9\%\\
\includegraphics[width=0.67in, height=0.17in]{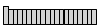} & Uniform & 79.8\% & 73.2\% & 79.5\% & 79.9\%\\
\includegraphics[width=0.67in, height=0.17in]{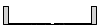} & Beta & 78.1\% & 68.8\% & 79.3\% & 79.9\%\\
\includegraphics[width=0.67in, height=0.17in]{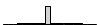} & Sparse 3 & 85.0\% & 82.6\% & 81.9\% & 81.9\%\\
\includegraphics[width=0.67in, height=0.17in]{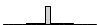} & Sparse 2 & 87.5\% & 85.6\% & 84.5\% & 84.4\%\\
\includegraphics[width=0.67in, height=0.17in]{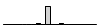} & Sparse 1 & 92.6\% & 91.0\% & 90.3\% & 90.4\%\\
\bottomrule
\end{tabular}
  \caption{\label{Tab:dist_beta} Results are from the simulation described in Section~\ref{Sec:distribution}. The nominal coverage rate is 80\%.}
\end{table}

\subsubsection{Selection of \texorpdfstring{$\lambda$}{lambda}} \label{Sec:lambda}

Throughout the manuscript, $\lam$ is set to the value that minimizes CV error; here, we examine how the choice of $\lam$ affects coverage. The design remains the same as in Section~\ref{Sec:coverage} except that for each data set generated, RL-P intervals are obtained for 25 different values of $\lambda$. Specifically, $\lambda$ was evenly distributed on the $\log_{10}$ scale from $\lam_{\max}$ to $\lam_{\min} = 0.05 \lam_{\max}$. At each value, confidence intervals were obtained and coverage was recorded. This was repeated 1000 times to estimate coverage as a function of $\lambda$ and $|\beta|$. Relative coverage is defined here as the estimated coverage rate minus the nominal coverage rate (red values denote coverage less than nominal, blue values above nominal). For example, at a nominal coverage rate of 80\%, if coverage for a given combination of $|\beta|$ and $\lam$ is estimated to be 55\%, this leads to a relative coverage of 55\% - 80\% = -25\%. The x-axis for $\lambda$ is presented relative to $\lam_{\max}$ and the solid black lines delineate the center 75\% of $\lam_{\CV}$ over the 1000 simulations. The dashed black line indicates the median $\lambda_{\CV}$ and the blue line represents the value of $\lambda$ which provided coverage closest to that of nominal.

\begin{figure}[htb!]
  \begin{center}
    \includegraphics[width=0.6\linewidth]{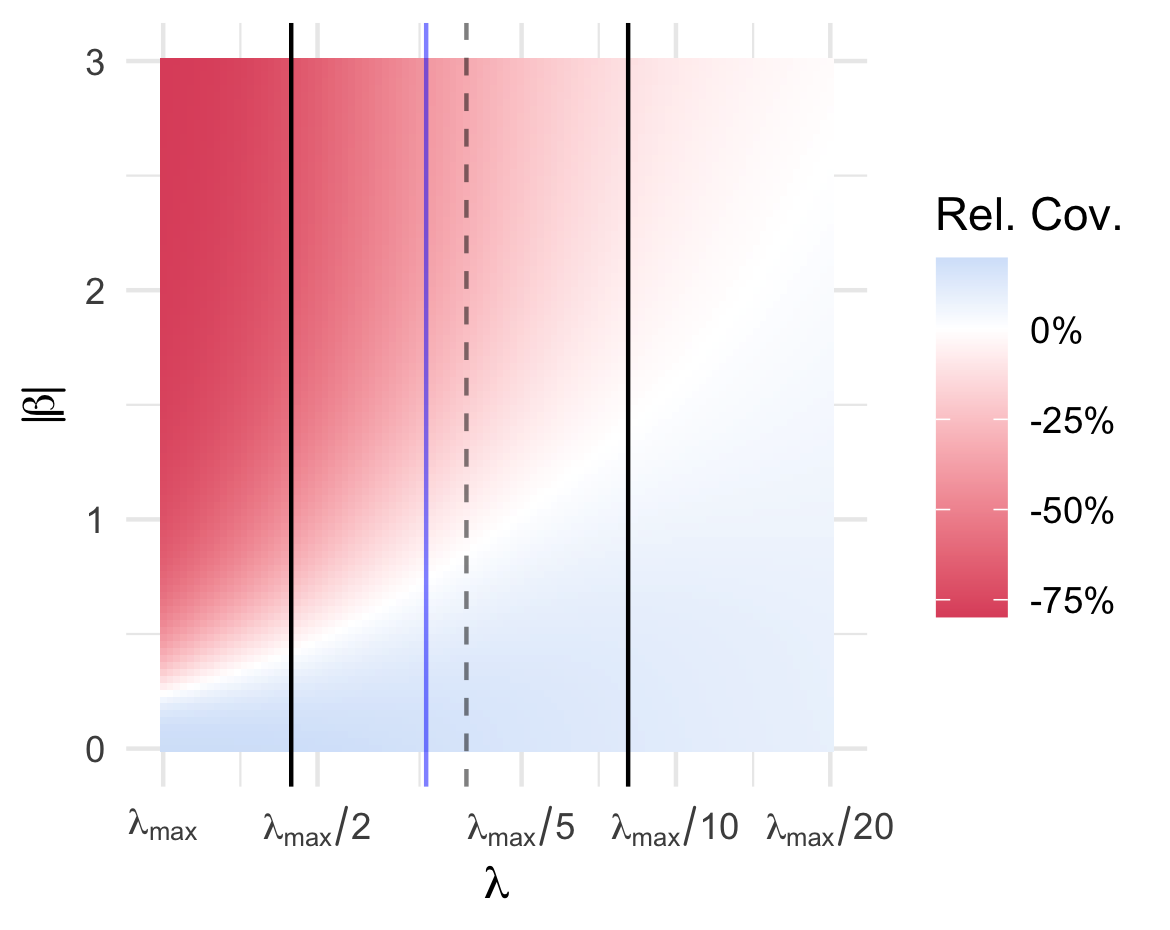}
    \caption{\label{Fig:beta_lambda_heatmap_laplace} The heatmap displays relative coverage for RL-P across a range of $\lambda$s per the simulation described in Section~\ref{Sec:lambda}. A Binomial GAM was used to estimate coverage as a smooth function of the $|\beta|$ and $\lam$. The x-axis for $\lambda$ is presented relative to $\lam_{\max}$ and the solid black lines indicate the center 75\% of $\lam_{\CV}$s over the 1000 simulations. The dashed black line indicates the median $\lambda_{\CV}$ and the blue line represents the value of $\lambda$ which provided coverage closest to that of nominal.}
  \end{center}
\end{figure}

To obtain average coverage near nominal, $\lam$ should be chosen such that the over-coverage for small $|\beta|$ values is balanced by the under-coverage for large $|\beta|$ values. The blue line in Figure~\ref{Fig:beta_lambda_heatmap_laplace} represents the $\lam$ value for which this balance is best achieved. In this scenario, and in general, $\lam_{\CV}$ does a reasonable job at achieving this balance: sometimes below the ``perfect balance'' line, sometimes above, but usually in reasonable agreement.

However, clearly the value of $\lam$ does matter, again supporting the idea presented in Theorem~\ref{Thm:bcc}. This simulation is similar to when $\bb$ is distributed as alternative distributions, but here, instead of altering the data generating mechanism, we are adjusting the prior implied by the lasso penalty. When this implied prior is reasonably close to the data generating mechanism, coverage is near nominal. When the implied prior is more concentrated at zero (e.g. $\lambda$ near $\lambda_{\max}$) or more diffuse (e.g. $\lambda$ near $\lambda_{\min}$), then coverage is below and above nominal, respectively.

\subsection{Effect of Correlation on Individual Intervals} \label{Sec:Ridge}

As shown in Figure~\ref{Fig:correlation_structure}, the average coverage of the proposed RL-P method is robust to increasing correlation. However, this does not mean that the intervals themselves are unaffected by correlation. In this section, we illustrate the effect of correlation between features on the intervals themselves and contrast the intervals produced by lasso with those produced by ridge regression.

In this simulation, we have $n = p = 100$. However, only one $\beta_j$ is non-zero: $\beta_{A} = 1$ and $\beta_{B}, \beta_{N1}, \ldots, \beta_{N98} = 0$. Additionally, the data are simulated such that $\cor(\x_{A}, \x_{B}) = .99$ but all of the N (noise) variables are uncorrelated with $A$, $B$, and each other. The distribution of $\X$ and $\y$ is unchanged from Section~\ref{Sec:coverage}, although here $\sigma^2 = 1$.

Figure~\ref{Fig:highcorr} depicts the results from $1,000$ simulated data sets. On top, 1000 CIs are shown for for 3 features: $A$, $B$, and $N1$; the CIs are colored black if they contain the true parameter value and red if they do not. On bottom, confidence intervals for the first 20 variables for a randomly selected example data set are displayed.

\begin{figure}[htb!]
  \begin{center}
    \includegraphics[width=0.8\linewidth]{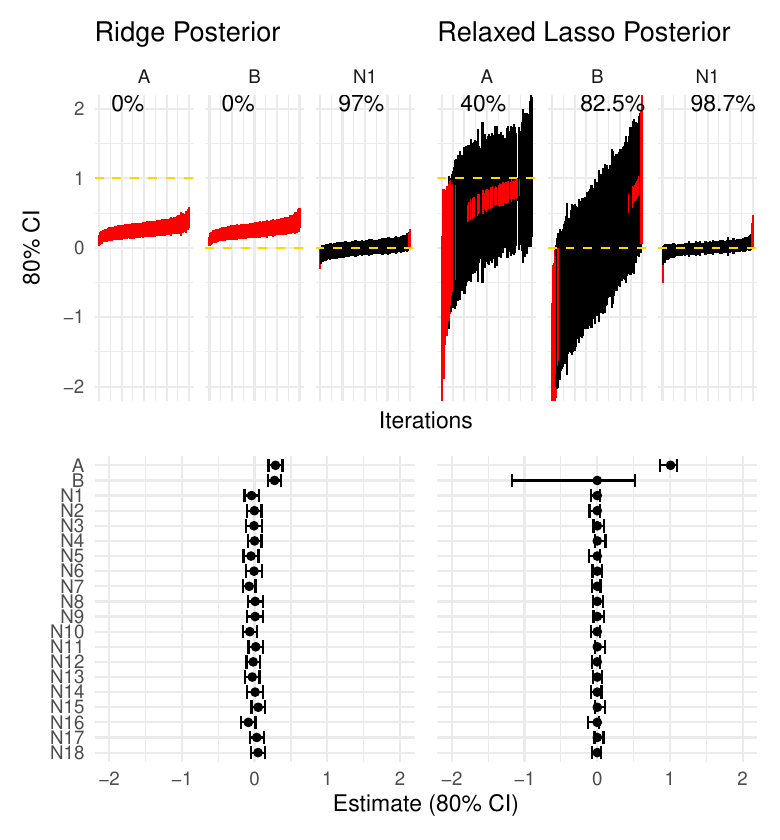}
    \caption{\label{Fig:highcorr}
      Provides results for simulation described in Section~\ref{Sec:Ridge}. The bottom plots show a single example of a intervals produced by Ridge (left) and RL-P (right) from one (randomly selected) of the 1000 datasets for the first 20 variables. The top plot summarizes the resulting CIs for the variables $A$, $B$, and $N1$ across the 1000 simulations. All 1000 CIs are plotted, sorted by their midpoint, with those colored red that did not contain contain the true coefficient value (indicated by the horizontal dashed gold line).
    }
  \end{center}
\end{figure}

Although only feature $A$ truly carries the signal, its high correlation with feature $B$ makes it unclear whether the signal originates solely from $A$ or from both $A$ and $B$. Ridge and Lasso resolve this ambiguity in different ways. Ridge regression makes a fairly strong assumption here that it is much more likely for the signal to be divided equally between $A$ and $B$ than for either $A$ or $B$ to have all the signal. This results in intervals that are very similar for $A$ and $B$. As a result, the correlation between $A$ and $B$ does not introduce much uncertainty -- the confidence intervals for $\beta_A$ and $\beta_B$ are no wider than that of the noise features.

With the Lasso on the other hand, the following scenarios are all equally penalized: $A$ has all the signal, $B$ has all the signal, and the signal is shared between $A$ and $B$. As a result, Lasso estimates are very sensitive to correlation and accordingly, the RL-P CIs for $A$ and $B$ are typically much wider than those for the noise features. Furthermore, although the width of the intervals are similar, the RL-P intervals for $A$ tend to be shifted towards higher values compared to those for $B$, indicating that even with very high correlation, the Lasso typically attributes more of the signal to the causal feature $A$; this is not the case with the Ridge penalty. Lastly, note that the width of the RL-P intervals is bimodal: either narrow or very wide. This can be seen from the construction in Section~\ref{Sec:methods}: the variance of the conditional distribution is largely determined by how much information in $\x_j$ is orthogonal to $\X_{\hat{S}_j}$. When variable $A$ is selected, the interval for variable $B$ will be wide and vice versa.

\subsection{Comparison to other methods} \label{Sec:Comparison}

As noted in the introduction, there are few methods for obtaining intervals for the lasso that have been developed and implemented with available software. Two that we were able to identify were Selective Inference (implemented in the \texttt{selectiveInference} R package) and the de-sparsified lasso (implemented in the \texttt{hdi} R package).

Selective Inference, de-sparsified lasso, and RL-P are based on different principles and operate in fundamentally distinct ways. To review, the de-sparsified lasso \citep{ZhangZhang2014}, as the name suggests, provides a method to debias the original point estimates from a lasso fit to facilitate classical approaches to inference. Note that this process of debiasing changes the underlying model, a point we return to in Section~\ref{Sec:discussion}. Alternatively, Selective Inference \citep{Lee2016,Tibshirani2016} aims to account for the uncertainty in model selection by conditioning on the selected model. This conditioning also acts to correct bias, albeit indirectly. Note that by conditioning on the selected model, Selective Inference only directly provides intervals for the covariates that were selected.

\begin{figure}[htb!]
  \begin{center}
    \begin{minipage}[t]{0.34\linewidth}
      \centering
      \includegraphics[width=\linewidth]{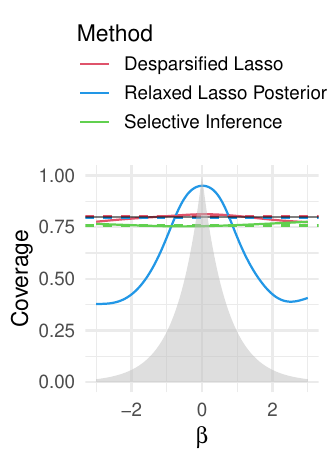}
    \end{minipage}%
    \hfill
    \begin{minipage}[t]{0.64\linewidth}
      \centering
      \includegraphics[width=\linewidth]{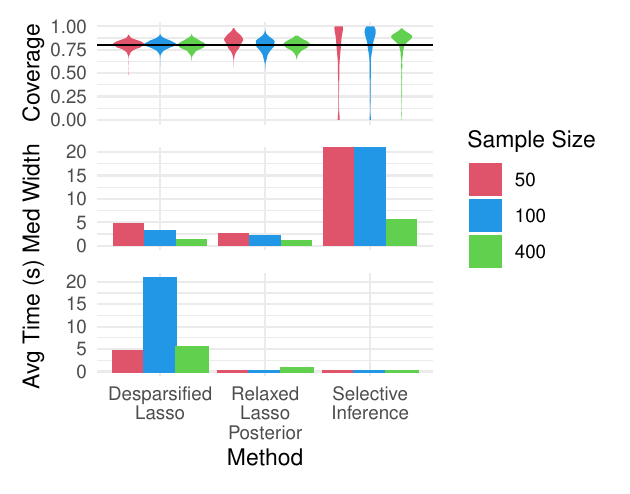}
    \end{minipage}
    \caption{\label{Fig:laplace_comparison} Results are from the simulation described in Section~\ref{Sec:Comparison}. On the left, the fitted curves are from Binomial GAMs fit with coverage being modeled as a smooth function of $\beta$. The dashed lines represent the average coverage for each method across all 1000 independently generated datasets and the solid black line indicates the nominal coverage rate. The shaded background is the distribution of $\beta$ (Laplace). On the right, each plot provides corresponding results for each of de-sparsified lasso, RL-P, and Selective Inference all three different sample sizes. The top provides violin plots of average coverages, the middle is a bar plot of the the median CI widths, and the bottom is a bar plot of the average run times, across all 1000 simulated datasets. The y limits have been truncated for the median width from 150 to 20.}
  \end{center}
\end{figure}

We conducted a simulation study to compare these three methods; the setup is identical to that described in Section~\ref{Sec:coverage}. For each software package, their default options were used. Notably, this means that for de-sparsified lasso's implementation in the \texttt{hdi} package, $\lam$ is set using the 1SE rule from cross-validation, whereas for Selective Inference and RL-P, $\lam$ was set at the value which minimizes CV error.

Selective Inference and de-sparsified lasso adopt a more classical frequentist perspective than RL-P, which is evident in the left side of Figure~\ref{Fig:laplace_comparison}. While all three methods have reasonable average coverage, they achieve this in different ways. As demonstrated in Section~\ref{Sec:IAC}, methods for constructing intervals can either achieve consistent coverage across all values of the target parameter, or they can reflect the shrinkage imposed by the penalty --- they cannot achieve both. Intervals which reflect the shrinkage imposed by the penalty result in uneven coverage across $\beta$. As shown in left left side of Figure~\ref{Fig:laplace_comparison}, Selective Inference and de-sparsified lasso provide the first kind of interval. Either directly or indirectly, the shrinkage imposed by the lasso has been undone by intervals they provide and the result is flat coverage across values of $\beta$. This is unlike RL-P, which reflects the shrinkage of the lasso and results in higher coverage where the prior density (implied by the penalty) is higher.

The right side of Figure~\ref{Fig:laplace_comparison} illustrates how the coverage, interval width, and computational burden of these methods compare. The top panel shows the distribution of average coverage across all 1000 simulations. The average coverage of the de-sparsified lasso is centered around the nominal 80\% coverage. Meanwhile, the distribution of average coverage for Selective Inference is very wide: for some data sets, average coverage was 100\% while for other data sets average coverage was 0\%. For $n = 50$ and $n = 100$, although centered around nominal coverage, average coverage was often well above or well below the nominal rate. The average coverage is less variable at $n = 400$, although it is remains consistently above the nominal rate with a noticeable tail down to 0\%. The behavior for RL-P has been covered previously: it is slightly conservative when $n$ is small but converges to nominal coverage as $n$ increases.

The middle plot provides the median CI width across all covariates from all 1000 simulations. The de-sparsified lasso tends to produce wider intervals, especially when $p \le n$, than RL-P. Selective Inference, on the other hand, produces much wider intervals than the other two methods. In fact, the vertical limits of the panel had to be truncated -- to capture the full bar for Selective Inference when $n = 50$, the vertical limit would need to go up to 150.

Selective Inference differs from de-sparsified lasso and RL-P in that it does not provide intervals for all parameters, only the subset of parameters with nonzero coefficients. And even when Selective Inference does construct intervals, they are often infinitely wide (we will see this again for the real data in Section~\ref{Sec:RDA}). More information on how often these two issues arise in this simulation is found in \ref{Sup:si_int_info}. \citet{Kivaranovic2021} provide an in-depth discussion of the widths of CIs produced by methods like Selective Inference that use a polyhedral approach and show that the expected value of interval width is infinite.

The bottom panel of the right side of Figure~\ref{Fig:laplace_comparison} provides the average run times for each of the methods. The runtime varied considerably between the methods, with Selective Inference the fastest and de-sparsified lasso by far the slowest.  The only noticeable difference in speeds between RL-P and Selective Inference is that Selective Inference scales better with n. In our testing, speed was not a concern for Selective Inference or RL-P, but the de-sparsified lasso was prohibitively slow. Although not shown in the figure, de-sparsified lasso also scales quite poorly with $p$, as we will see in Section~\ref{Sec:Scheetz2006}.

\section{Applications to Real Data}\label{Sec:RDA}

In this section, we apply the RL-P method to two real datasets: a study of acute respiratory illness conducted by the World Health Organization, and a study of gene expression in the mammalian eye \citep{Scheetz2006}. These two datasets sit on opposite ends of the spectrum in terms of dimensionality. The WHO study contains 816 observations and 66 features, while the gene expression study has just 120 observations and 18,975 features. In this section, we also consider the intervals produced by de-sparsified lasso and Selective Inference, comparing the intervals both to each other and to the point estimates provided by the lasso at $\lam_{\CV}$.

\subsection{World Health Organization study on acute respiratory illnesses}\label{Sec:WHO-ARI}

The World Health Organization/Acute Respiratory Infection (WHO/ARI) Multicentre Study collected data on several acute respiratory illnesses in multiple countries \citep{Harrell1998}. Here, we analyze a subset of this study, concerning 816 infants who presented with symptoms of pneumonia in Ethiopia, a major cause of morbidity and mortality for infants under 3 months of age. Our goal here is to identify risk factors for increased severity among infants presenting with serious infections. The outcome is ordinal (taking on a number from 1 - 5); however, for simplicity we treat the outcome as following a Gaussian distribution. The variables collected contain information on vital signs, family history, and clinical observations and represent a range of data types from binary to ordinal to continuous. With $n \approx 10p$, this dataset is not high dimensional. However, sparsity is beneficial both for interpretation and for the practical implementation of using the resulting model in clinical practice.

\begin{figure}[htb!]
  \begin{center}
    \includegraphics[width=0.9\linewidth]{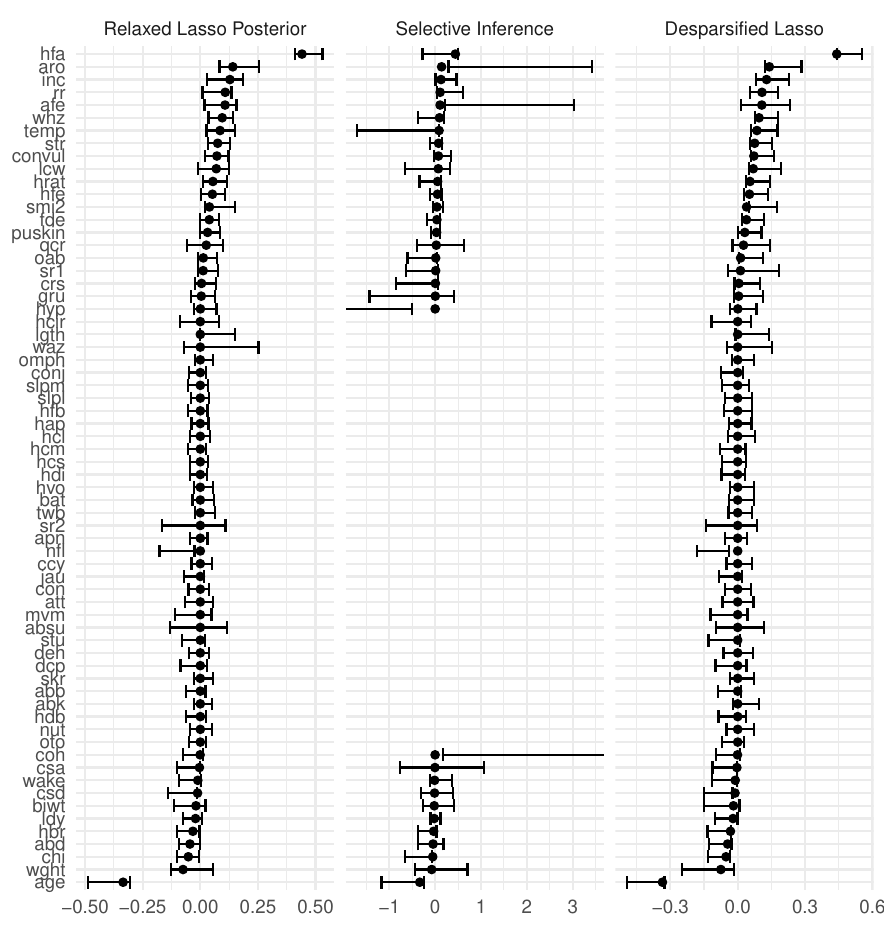}
    \caption{\label{Fig:comparison_data_whoari} Confidence intervals produced by three different methods for all 66 variables in the WHO/ARI dataset described in Section~\ref{Sec:WHO-ARI}.}
  \end{center}
\end{figure}

Figure~\ref{Fig:comparison_data_whoari} provides the confidence intervals from each of the three methods along with corresponding point estimates from the lasso. The intervals are provided on the standardized scale to aid in visualization. The RL-P and de-sparsified lasso intervals are generally similar, although the intervals from RL-P are narrower --- note that the horizontal axis range is different for each method. Additionally, while RL-P intervals, relative to the point estimates, tend to be more symmetric, de-sparsified lasso's intervals are more often skewed away from zero as a result of debiasing. As mentioned earlier, Selective Inference does not produce an interval for every parameter, only for the 32 (out of 66) features that were selected. Furthermore, of these 32, two intervals are infinitely wide and several others are much wider than any intervals produced by either de-sparsified lasso or RL-P. Altogether, de-sparsified lasso produces 25 intervals that do not contain zero, RL-P produces 19 that do not contain zero, while only 8 of the 32 Selective Inference intervals do not contain zero.

\subsection{Study of gene expression in the mammalian eye}\label{Sec:Scheetz2006}

\citet{Scheetz2006} measured the RNA levels from the eyes of 120 rats. Of 31,042 different probes used, 18,976 were detected at a sufficient level to be considered ``expressed.'' For this analysis we treat one of the genes, Trim32, as the outcome since it is known to be linked to the genetic disorder Bardet-Biedl Syndrome (BBS). The remaining 18,975 genes are used as covariates with the goal of determining other genes whose expression is associated with Trim32 and thus may also contribute to BBS.

\begin{figure}[htb!]
  \begin{center}
    \includegraphics[width=0.9\linewidth]{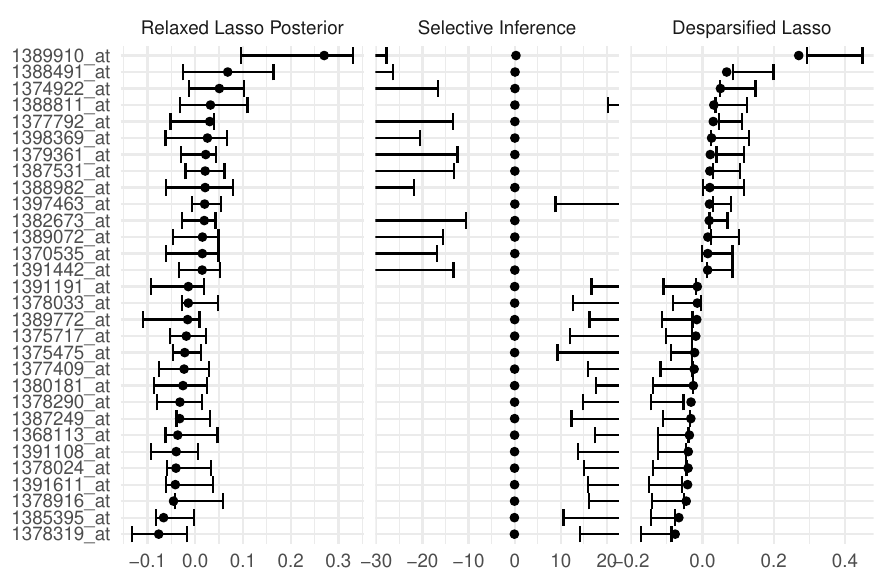}
    \caption{\label{Fig:comparison_data_scheetz} Confidence intervals produced by three different methods for the 30 variables with the largest absolute point estimates in the Scheetz2006 dataset described in Section~\ref{Sec:Scheetz2006}.}
  \end{center}
\end{figure}

Compared to the WHO/ARI data, the increased dimensionality here leads to more pronounced differences between RL-P and de-sparsified lasso (Figure~\ref{Fig:comparison_data_scheetz}). The RL-P intervals exhibit more shrinkage towards zero, while the intervals of de-sparsified lasso are pushed away from zero. Additionally, while 71 of the RL-P intervals do not contain their respective point estimates, this occurs for 981 intervals produced by the de-sparsified lasso. In addition, there is a large discrepancy for the number of significant intervals (intervals not containing zero) between the two methods. De-sparsified lasso produces 989 intervals which exclude zero, while RL-P produces 77. Selective Inference provides intervals for 66 of the 18975 features. For this high-dimensional data ($p > 100n$), however, every single one of them has a lower or upper bound that is infinite. Additionally, none of the Selective Inference intervals contain zero, and in fact, have no overlap with any of the de-sparsified lasso or RL-P intervals (note again that the horizontal axis is different for each of the three methods). Particularly troubling is the fact that of the 66 intervals created by Selective Inference, 62 of them were completely of the opposite sign as the corresponding lasso estimate.

Lastly, it is important to note that de-sparsified lasso took over 6 hours to produce these confidence intervals on a MacBook Pro with 16 GB of RAM and an Apple M1 Pro chip. This is because the computational cost of the de-sparsified lasso scales poorly with $p$. In comparison, RL-P took 1.2 seconds while Selective Inference took about three tenths of a second. With respect to computational burden, de-sparsified lasso is feasible for small to moderately sized datasets, but the cost becomes prohibitive when $p$ is large.

\section{Discussion} \label{Sec:discussion}

Should intervals be biased? Over the past several decades, statisticians have grown more comfortable with the idea of biased estimators. Nevertheless, the statistics community still appears to be uncomfortable with biased intervals. However, if you have chosen to use a biased estimation method, it would seem reasonable that the resulting intervals should reflect that bias. This conflicts with classical frequentist ideas of coverage, but as we have shown, agrees with Bayesian posterior intervals.

One objection to having biased intervals is that it results in under-coverage for large values of $\beta$ --- the parameters that are typically of greatest interest. The same objection, however, applies to the lasso estimates themselves. There are many alternatives to the lasso, including the adaptive lasso, MCP, and SCAD, which reduce the bias imposed by the lasso for large values of $\beta$ \citep{Zou2006, Zhang2010, Fan2001}. Using RL-P with any of these alternative approaches results in less biased intervals (\ref{Sup:MCP}).

In contrast, most of the literature on high-dimensional intervals focuses on debiased constructions. It is debatable, however, whether these intervals remain consistent with the assumptions underlying the original lasso model. Approaches such as the de-sparsified lasso and Selective Inference are not incorrect, but they do not reflect the original lasso estimates --- a point illustrated most clearly in Section~\ref{Sec:Scheetz2006}.

At a fundamental level, it is not possible for a statistical method to shrink point estimates towards zero while not also shrinking intervals towards zero. Attempting to accomplish both will lead to inconsistencies. For an analyst attempting to ``pair'' the lasso with de-sparsified lasso or Selective Inference, it is critical to recognize that the underlying assumptions for the point estimates and for the intervals are not the same. Presenting intervals and point estimates that do not agree is unsatisfying, leading to results that are difficult to interpret and ultimately less convincing.

Rather than try to debias or otherwise correct for the lasso penalty when constructing intervals, we develop here the Relaxed Lasso Posterior and show that it offers a more coherent approach where the intervals reflect the lasso point estimates. Adopting these intervals requires a change in perspective, where the emphasis of the intervals is average coverage across the set of parameters as opposed to individual parameter coverage. We think this work raises interesting questions about which perspective is preferable, and hope that it encourages further exploration of whether classical single-parameter frequentist ideas are the right foundation for high-dimensional inference.

\section*{Data Availability}

All datasets used are available in the R package \texttt{hdrm} (version 0.17.1) which can be found at \url{https://github.com/pbreheny/hdrm}.

\bibliographystyle{ims-nourl}

\begin{thebibliography}{26}
\expandafter\ifx\csname natexlab\endcsname\relax\def\natexlab#1{#1}\fi
\expandafter\ifx\csname url\endcsname\relax
  \def\url#1{\texttt{#1}}\fi
\expandafter\ifx\csname urlprefix\endcsname\relax\def\urlprefix{URL }\fi
\providecommand{\eprint}[2][]{\url{#2}}

\bibitem[{Barber and Cand{\`e}s(2015)}]{Candes2015}
\textsc{Barber, R.~F.} and \textsc{Cand{\`e}s, E.~J.} (2015).
\newblock Controlling the false discovery rate via knockoffs.
\newblock \textit{The Annals of Statistics}, \textbf{43} 2055--2085.

\bibitem[{Breheny(2019)}]{Breheny2019}
\textsc{Breheny, P.~J.} (2019).
\newblock Marginal false discovery rates for penalized regression models.
\newblock \textit{Biostatistics}, \textbf{20} 299--314.

\bibitem[{Cand{\`e}s et~al.(2018)Cand{\`e}s, Fan, Janson and Lv}]{Candes2018}
\textsc{Cand{\`e}s, E.}, \textsc{Fan, Y.}, \textsc{Janson, L.} and \textsc{Lv, J.} (2018).
\newblock Panning for gold: ‘model-x’ knockoffs for high dimensional controlled variable selection.
\newblock \textit{Journal of the Royal Statistical Society: Series B (Statistical Methodology)}, \textbf{80} 551--577.

\bibitem[{Chatterjee and Lahiri(2010)}]{Chatterjee2010}
\textsc{Chatterjee, A.} and \textsc{Lahiri, S.~N.} (2010).
\newblock Asymptotic properties of the residual bootstrap for lasso estimators.
\newblock \textit{Proceedings of the American Methematical Society}, \textbf{138} 4497--4509.

\bibitem[{Clart{\'e} et~al.(2024)Clart{\'e}, Vandenbroucque, Dalle, Loureiro, Krzakala and Zdeborov{\'a}}]{clarté2024}
\textsc{Clart{\'e}, L.}, \textsc{Vandenbroucque, A.}, \textsc{Dalle, G.}, \textsc{Loureiro, B.}, \textsc{Krzakala, F.} and \textsc{Zdeborov{\'a}, L.} (2024).
\newblock Analysis of bootstrap and subsampling in high-dimensional regularized regression.
\newblock In \textit{Proceedings of the 40th Conference on Uncertainty in Artificial Intelligence (UAI 2024)}, vol. 244. Proceedings of Machine Learning Research, 787--819.

\bibitem[{Efron(1982)}]{efron1982}
\textsc{Efron, B.} (1982).
\newblock \textit{The Jackknife, the Bootstrap and Other Resampling Plans}, vol.~38 of \textit{CBMS-NSF Regional Conference Series in Applied Mathematics}.
\newblock Society for Industrial and Applied Mathematics, Philadelphia, PA.

\bibitem[{Fan and Li(2001)}]{Fan2001}
\textsc{Fan, J.} and \textsc{Li, R.} (2001).
\newblock Variable selection via nonconcave penalized likelihood and its oracle properties.
\newblock \textit{Journal of the American Statistical Association}, \textbf{96} 1348--1360.

\bibitem[{Harrell et~al.(1998)Harrell, Margolis, Gove, Mason, Mulholland, Lehmann, Muhe, Gatchalian and Eichenwald}]{Harrell1998}
\textsc{Harrell, F.}, \textsc{Margolis, P.}, \textsc{Gove, S.}, \textsc{Mason, K.}, \textsc{Mulholland, E.}, \textsc{Lehmann, D.}, \textsc{Muhe, L.}, \textsc{Gatchalian, S.} and \textsc{Eichenwald, H.} (1998).
\newblock Development of a clinical prediction model for an ordinal outcome: the world health organization multicentre study of clinical signs and etiological agents of pneumonia, sepsis and meningitis in young infants.
\newblock \textit{Statistics in Medicine}, \textbf{17} 909--944.

\bibitem[{Hastie et~al.(2009)Hastie, Tibshirani and Friedman}]{HTF2009}
\textsc{Hastie, T.}, \textsc{Tibshirani, R.} and \textsc{Friedman, J.} (2009).
\newblock \textit{The Elements of Statistical Learning: Data Mining, Inference, and Prediction}.
\newblock 2nd ed. Springer Series in Statistics, Springer.

\bibitem[{Javanmard and Montanari(2014)}]{Javanmard2014}
\textsc{Javanmard, A.} and \textsc{Montanari, A.} (2014).
\newblock Confidence intervals and hypothesis testing for high-dimensional regression.
\newblock \textit{Journal of Machine Learning Research (JMLR)}, \textbf{15} 2869--2909.

\bibitem[{Karoui and Purdom(2018)}]{karoui2016}
\textsc{Karoui, N.~E.} and \textsc{Purdom, E.} (2018).
\newblock Can we trust the bootstrap in high-dimension? the case of linear models.
\newblock \textit{Journal of Machine Learning Research}, \textbf{19} 1--66.

\bibitem[{Kivaranovic and Leeb(2021)}]{Kivaranovic2021}
\textsc{Kivaranovic, D.} and \textsc{Leeb, H.} (2021).
\newblock On the length of post-model-selection confidence intervals conditional on polyhedral constraints.
\newblock \textit{Journal of the American Statistical Association}, \textbf{116} 845--857.

\bibitem[{Lee et~al.(2016)Lee, Sun, Sun and Taylor}]{Lee2016}
\textsc{Lee, J.~D.}, \textsc{Sun, D.~L.}, \textsc{Sun, Y.} and \textsc{Taylor, J.~E.} (2016).
\newblock Exact post-selection inference, with application to the lasso.
\newblock \textit{The Annals of Statistics}, \textbf{44} 907--927.

\bibitem[{Lo(1987)}]{Lo1987}
\textsc{Lo, A.~Y.} (1987).
\newblock A large sample study of the bayesian bootstrap.
\newblock \textit{The Annals of Statistics}, \textbf{15} 360--375.

\bibitem[{Lockhart et~al.(2014)Lockhart, Taylor, Tibshirani and Tibshirani}]{Lockhart2014}
\textsc{Lockhart, R.}, \textsc{Taylor, J.}, \textsc{Tibshirani, R.~J.} and \textsc{Tibshirani, R.} (2014).
\newblock A significance test for the lasso.
\newblock \textit{The Annals of Statistics}, \textbf{42} 413--468.

\bibitem[{Meinshausen and B{\"u}hlmann(2010)}]{Meinshausen2010}
\textsc{Meinshausen, N.} and \textsc{B{\"u}hlmann, P.} (2010).
\newblock Stability selection.
\newblock \textit{Journal of the Royal Statistical Society: Series B (Statistical Methodology)}, \textbf{72} 417--473.

\bibitem[{Park and Casella(2008)}]{Park2008}
\textsc{Park, T.} and \textsc{Casella, G.} (2008).
\newblock The bayesian lasso.
\newblock \textit{Journal of the American Statistical Association}, \textbf{103} 681--686.

\bibitem[{Reid et~al.(2016)Reid, Tibshirani and Friedman}]{Reid2016}
\textsc{Reid, S.}, \textsc{Tibshirani, R.} and \textsc{Friedman, J.} (2016).
\newblock A study of error variance estimation in lasso regression.
\newblock \textit{Statistica Sinica}, \textbf{26} 35--67.

\bibitem[{Rubin(1981)}]{Rubin1981}
\textsc{Rubin, D.~B.} (1981).
\newblock {The Bayesian Bootstrap}.
\newblock \textit{The Annals of Statistics}, \textbf{9} 130 -- 134.

\bibitem[{Scheetz et~al.(2006)Scheetz, Kim, Swiderski, Philp, Braun, Knudtson, Dorrance, DiBona, Huang, Casavant, Sheffield and Stone}]{Scheetz2006}
\textsc{Scheetz, T.}, \textsc{Kim, K.-Y.}, \textsc{Swiderski, R.}, \textsc{Philp, A.}, \textsc{Braun, T.}, \textsc{Knudtson, K.}, \textsc{Dorrance, A.}, \textsc{DiBona, G.}, \textsc{Huang, J.}, \textsc{Casavant, T.}, \textsc{Sheffield, V.} and \textsc{Stone, E.} (2006).
\newblock Regulation of gene expression in the mammalian eye and its relevance to eye disease.
\newblock \textit{Proceedings of the National Academy of Sciences}, \textbf{103} 14429--14434.

\bibitem[{Tibshirani(1996)}]{Tibshirani1996}
\textsc{Tibshirani, R.} (1996).
\newblock Regression shrinkage and selection via the lasso.
\newblock \textit{Journal of the Royal Statistical Society Series B}, \textbf{58} 267--288.

\bibitem[{Tibshirani et~al.(2016)Tibshirani, Taylor, Lockhart and Tibshirani}]{Tibshirani2016}
\textsc{Tibshirani, R.~J.}, \textsc{Taylor, J.}, \textsc{Lockhart, R.} and \textsc{Tibshirani, R.} (2016).
\newblock Exact post-selection inference for sequential regression procedures.
\newblock \textit{Journal of the American Statistical Association}, \textbf{111} 600--620.

\bibitem[{Xing et~al.(2023)Xing, Zhao and Liu}]{Xing2023}
\textsc{Xing, X.}, \textsc{Zhao, Z.} and \textsc{Liu, J.~S.} (2023).
\newblock Controlling false discovery rate using gaussian mirrors.
\newblock \textit{Journal of the American Statistical Association}, \textbf{118} 222--241.

\bibitem[{Zhang(2010)}]{Zhang2010}
\textsc{Zhang, C.-H.} (2010).
\newblock Nearly unbiased variable selection under minimax concave penalty.
\newblock \textit{The Annals of Statistics}, \textbf{38} 894--942.

\bibitem[{Zhang and Zhang(2014)}]{ZhangZhang2014}
\textsc{Zhang, C.~H.} and \textsc{Zhang, S.~S.} (2014).
\newblock Confidence intervals for low dimensional parameters in high dimensional linear models.
\newblock \textit{Journal of the Royal Statistical Society: Series B (Statistical Methodology)}, \textbf{76} 217--242.

\bibitem[{Zou(2006)}]{Zou2006}
\textsc{Zou, H.} (2006).
\newblock The adaptive lasso and its oracle properties.
\newblock \textit{Journal of the American Statistical Association}, \textbf{101} 1418--1429.

\end{thebibliography}

\clearpage

\begin{appendices}
\renewcommand{\thesection}{Appendix~\Alph{section}}
\renewcommand{\thefigure}{Figure~\Alph{section}\arabic{figure}}
\renewcommand{\thetable}{Table~\Alph{section}\arabic{table}}
\renewcommand{\thealgorithm}{Algorithm~\Alph{section}\arabic{algorithm}}

\makeatletter
\@addtoreset{figure}{section}
\@addtoreset{table}{section}
\@addtoreset{algorithm}{section}
\makeatother

\captionsetup[figure]{name={}}
\captionsetup[table]{name={}}
\floatname{algorithm}{} 
  
\section{Traditional Bootstrap Example}
\label{sec:boot-fail}

Figure~\ref{Fig:boot-fail} shows the coverage as a function of the value of $\beta$ (solid line) and the average coverage (dashed line) using a traditional pairs bootstrapping approach on the simulation setup described in Section~\ref{Sec:coverage}. Compare to the right side of Figure~\ref{Fig:laplace}.

\begin{figure}[hbtp]
  \begin{center}
  \includegraphics[width=0.7\linewidth]{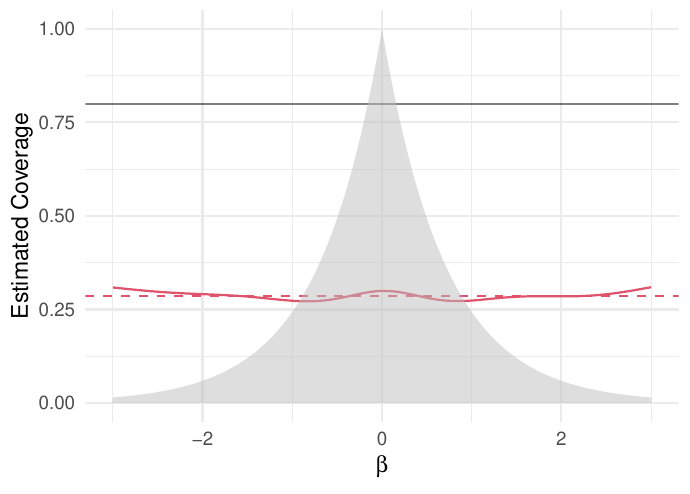}
  \caption{\label{Fig:boot-fail} Corresponds to the setup used in the right side of Figure~\ref{Fig:laplace}, but using a traditional bootstrapping approach.}
  \end{center}
\end{figure}

\clearpage

\section{Sampling from the Full Conditional Posterior}\label{Sup:A}

Here we assume $\X$ has been standardized s.t. $\x_j^T\x_j = n$. Define $\Q_{\hat{S}_j}$ as $\I - \X_{\hat{S}_j}(\X_{\hat{S}_j}^T \X_{\hat{S}_j})^{-1} \X_{\hat{S}_j}^T$, the projection matrix onto the features selected by the lasso. For $\beta_j$ conditional on the selected features

\as{
  \begin{aligned}
  L(\beta_j|\hat{S}_j) &\propto \exp(-\frac{\tilde{n}}{2\sigma^2}(\beta_{j}^2 - 2\tilde{\beta}_{j}\beta_{j})), \\
  \end{aligned}
}

\noindent where $\tilde{\beta}_j = (\x_j^T \Q_{\hat{S}_j} \x_j)^{-1} \x_j^T \Q_{\hat{S}_j} \y$ and $\tilde{n} = \x_j^T \Q_{\hat{S}_j} \x_j$.

The lasso can be formulated as a Bayesian regression model with a laplace (double exponential) prior. In this case, the prior for $\beta_j$ is proportional to $\exp(-\frac{\tilde{n} \lambda} {\sigma^2} \abs{\beta_j})$. This prior ensures that the meaning of $\lam$ is maintained. 

With this the form of the full conditional posterior can be worked out as follows:
\as{
p(\beta_j | \hat{S}_j) &\propto \exp(-\frac{\tilde{n}}{2\sigma^2} (\beta_j^2 - 2\tilde{\beta}_j\beta_j)) \exp(-\frac{\tilde{n} \lambda} {\sigma^2} \abs{\beta_j}) \\
&= \exp(-\frac{\tilde{n}}{2\sigma^2} (\beta_j^2 - 2 \tilde{\beta}_{j}\beta_j +  2 \lambda \abs{\beta_j})) \\
&= \exp(-\frac{\tilde{n}}{2\sigma^2} (\beta_j^2 - 2(\tilde{\beta}_{j}\beta_j - \lambda \abs{\beta_j}))) \\
&=
\begin{cases}
\exp(-\frac{\tilde{n}}{2\sigma^2} (\beta_j^2 - 2(\tilde{\beta}_{j} + \lambda)\beta_j)), \text{ if } \beta_j < 0, \\
\exp(-\frac{\tilde{n}}{2\sigma^2} (\beta_j^2 - 2(\tilde{\beta}_{j} - \lambda)\beta_j )), \text{ if } \beta_j \geq 0 \\
\end{cases} \\
&\propto
\begin{cases}
C_{-} \exp\{-\frac{\tilde{n}}{2\sigma^2} (\beta_j - (\tilde{\beta}_j + \lambda))^2\}, \text{ if } \beta_j < 0, \\
C_{+} \exp\{-\frac{\tilde{n}}{2\sigma^2} (\beta_j - (\tilde{\beta}_j - \lambda))^2\}, \text{ if } \beta_j \geq 0 \\
\end{cases}
}
where $C_{-} = \exp(\tilde{\beta}_j \lambda \tilde{n}/\sigma^2)$ and $C_{+} = \exp(-\tilde{\beta}_j \lambda \tilde{n}/\sigma^2)$.

Note the piecewise defined posterior is made up of a kernel of two normal distributions. This can be leveraged and draws can be efficiently obtained through a mapping onto the respective normal distributions. To define this mapping, it helps to introduce a concept and some notation. First, the use of ``tails'' here refers to the entirety of a distribution between 0 and $\pm \infty$. That is, the lower tail is any part of the distribution below 0 and the upper tail is any part greater than 0, therefore $P(X \in lower \cup X \in upper) = 1$. Accordingly, we will let the tail probabilities in each of the two normals to transformed on to be denoted $Pr_{-}$ and $Pr_{+}$ respectively and the probability in each of the tails of the posterior, denoted $Post_{-}$ and $Post_{+}$ respectively. $Pr_{\pm}$ is trivial to compute with any statistical software. $Post_{\pm}$ is conceptually simple, although care must be taken to avoid numerical instability. With this notation in place, note that,
\as{
p(\beta_j | \hat{S}_j)  & \propto
\begin{cases}
C_{-} Pr_{-}, \text{ if } \beta_j < 0, \\
C_{+} Pr_{+}, \text{ if } \beta_j \geq 0\\
\end{cases}
} which implies that $Post_- = \frac{C_{-} Pr_{-}}{C_{-} Pr_{-} + C_{+} Pr_{+}}$ and similarly for $Post_+$. To avoid numerical instability, or rather to handle it properly when it is unavoidable, we will work on the $\log$ scale. This works well for most of the problem, but computation of $Post_-$ and $Post_+$ need something a bit more since, for example, $\log(Post_-) = \log(C_{-}Pr_{-}) - \log(C_{-} Pr_{-} + C_{+} Pr_{+})$. That is, the denominator still must be computed then the $\log$ taken which does not allow operating on the $\log$ scale to fully address potential numerical instability. Instead, let $\ell_{-} := \log C_{-} + \log Pr_{-}, \ell_{+} := \log C_{+} + \log Pr_{+}$, and $\Delta := \ell_{-} - \ell_{+}$, then $\log(Post_-)$ can be computed with $\Delta - \log\bigl(1 + \exp(\Delta)\bigr)$. This still doesn't completely address the issue, however, if $\exp(\Delta)$ is infinite then $C_-Pr_- >> C_+Pr_+$ and $\log(Post_-) \approx 0$ which means $Post_- \approx 1$ (equivalently $Post_+ \approx 0$).

With these values, we can compute the quantile by mapping the corresponding probability $p$ for the posterior onto the probability $p^*$ for the corresponding normals. Which normal the quantile of interest ultimately comes from is determined based on $Post_{\pm}$. If $p \leq Post_{-}$, then $p$ would be mapped onto the negative normal. If $p > Post_{-}$, then $p$ would be mapped onto the positive normal.  For example, if $Post_{+} = 0.98$ and $p = 0.1$ then $p$ would be mapped onto the positive normal. The transformation to map a given probability from the posterior depends on which tail the quantile resides in on the posterior (equivalently which normal it is being mapped to, the positive or negative). This map is simply:

\as{
p^* &= p \times (Pr_{\pm} / Post_{\pm}) \\
}

Once the respective probability is mapped, one can simply use the inverses of the normal CDF that the probability was mapped to. That being said, there is a nuance worth pointing out. When transforming the probabilities, the step to determine which tail the respective quantile comes from occurs first. With this, the probability should be adjusted so that it refers to the probability between the quantile of interest and the respective tail. After this, then the transformation can be applied. These steps are summarized in Algorithm~\ref{alg:quantile}.

\begin{algorithm}
\caption{Compute Quantile for RL-P Laplace-Normal Distribution}
\label{alg:quantile}
\begin{algorithmic}[1]
  \Require 
    $\tilde{\beta}_j,\;\sigma^2,\;\tilde{n};\lambda,\;p$ \Comment{mean, variance, sample size, penalty, target significance level}
  \Statex
  \State // 1. Compute prior mass for negative and positive regions (on log‐scale)
  \State $Pr_{-} \gets \Phi\bigl(0;\,\tilde{\beta}_j + \lambda,\;\tfrac{\sigma^2}{\tilde{n}}\bigr)$
  \State $Pr_{+} \gets 1 - \Phi\bigl(0;\,\tilde{\beta}_j - \lambda,\;\tfrac{\sigma^2}{\tilde{n}}\bigr)$
  \Statex
  \State // 2. Compute posterior weights $Post_{-},\,Post_{+}$ with log‐scale stabilization
  \State $\ell_{-} \gets \log C_{-} + \log Pr_{-}$
  \State $\ell_{+} \gets \log C_{+} + \log Pr_{+}$
  \State $\Delta \gets \ell_{-} - \ell_{+}$
  \If{$\exp(\Delta) = \infty$}
    \State $\log Post_{-} \gets 0$  \Comment{since $C_-Pr_- \gg C_+Pr_+$}
  \Else
    \State $\log Post_{-} \gets \Delta - \log\bigl(1 + \exp(\Delta)\bigr)$
  \EndIf
  \State $Post_{-} \gets \exp(\log Post_{-})$
  \State $Post_{+} \gets 1 - Post_{-}$
  \Statex
  \State // 3. Invert CDF on the appropriate component
  \If{$p \le Post_{-}$}
    \State $w \gets \dfrac{Pr_{-}}{Post_{-}}$
    \State $q \gets \Phi^{-1}\bigl(p \,w;\;\tilde{\beta}_j + \lambda,\;\tfrac{\sigma^2}{n}\bigr)$
  \Else
    \State $w \gets \dfrac{Pr_{+}}{Post_{+}}$
    \State $q \gets \Phi^{-1}\!\Bigl(1 - (1-p)\,w;\;\tilde{\beta}_j - \lambda,\;\tfrac{\sigma^2}{n}\Bigr)$
  \EndIf
  \State \Return $q$
\end{algorithmic}
\end{algorithm}

\clearpage

\section{Coverage Behavior Under Alternative Sample Sizes}\label{Sup:alt_ns}

Figure~\ref{Fig:coverage_by_n} displays coverage estimates as a smooth function of $\beta$ for three values of n: 50, 100, and 400 but otherwise uses the same setup as the simulation described in Section~\ref{Sec:coverage}.

\begin{figure}[hbtp]
  \begin{center}
  \includegraphics[width=0.6\linewidth]{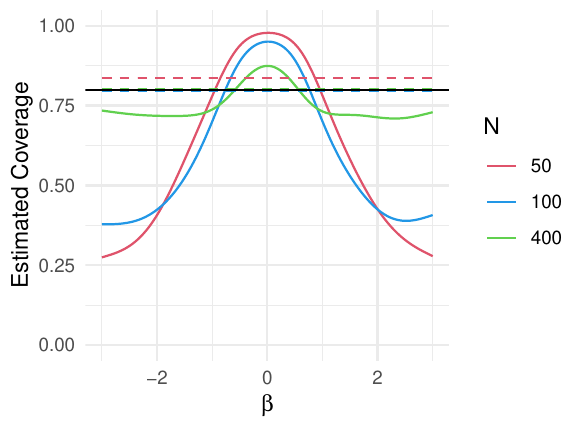}
  \caption{\label{Fig:coverage_by_n} The results displayed are from a simulation with the same set up as in Section~\ref{Sec:coverage} but with n set to three different values: 50, 100, 400. The fitted curves are from Binomial GAMs fit with coverage being modeled as a smooth function of $\beta$. The dashed lines represent the average coverages across all 1000 independently generated datasets and the solid black line indicates the 80\% nominal coverage rate.}
  \end{center}
\end{figure}

For $n = 50$, coverage is overconservative with the characteristic high coverage for values near zero and low coverage for larger values of $\beta$. However, As $n$ is increased to 100 then 400 the characteristic over coverage for small $\beta$ values and under coverage for large $\beta$ values lessens, largely attributable to $\lam_{\CV}$ being smaller. 

\clearpage

\section{Selective Inference Intervals}\label{Sup:si_int_info}

\begin{table}[hb]
  \singlespace
  \centering

\begin{tabular}[t]{cp{3cm}p{3cm}p{3cm}p{3cm}}
\toprule
n & \# Simulations Null Selected & Average \# Parameters & \# Simulations Inf Median Width & \# Simulations Any Inf Width\\
\midrule
50 & 660 & 13.033 & 127 & 349\\
100 & 274 & 28.539 & 125 & 571\\
400 & 0 & 73.441 & 39 & 626\\
\bottomrule
\end{tabular}
  \caption{Additional information on the results for Selective Inference in the simulation described in Section~\ref{Sec:Comparison}.}
  \label{Tab:selective_inference}
\end{table}

Because Selective Inference only provides intervals for the subset of parameters with nonzero coefficients, the results in Figure~\ref{Fig:laplace_comparison} are just for this subset. This is in contrast to the results for de-sparsified lasso and RL-P which are over all parameters. The first two columns of Table~\ref{Tab:selective_inference} provide more details on the size of the lasso models selected by cross validation, and hence, on the number of intervals constructed. The first column indicates how many simulations (out of 1000) the intercept-only model was selected, in which case no intervals are produced. The second column gives the average number of selected parameters (inclusive of when none were selected). It is the subset of parameters represented by column 2 that are the results in Figure~\ref{Fig:laplace_comparison} are from.

For the intervals that were constructed, the third and fourth columns provide the number of simulations that had a infinite median width or any interval with an infinite width, respectively, features that were not evident in the right hand side of Figure~\ref{Fig:laplace_comparison}. Note that simulations where the null model was chosen by definition can not have intervals of infinite width.

\clearpage

\section{MCP Example}\label{Sup:MCP}

\begin{table}[hbtp]
  \centering

\begin{tabular}[t]{ccc}
\toprule
\multicolumn{1}{c}{ } & \multicolumn{2}{c}{Coverage (\%)} \\
\cmidrule(l{3pt}r{3pt}){2-3}
$|\beta|$ & Relaxed Lasso Posterior & Relaxed MCP Posterior\\
\midrule
0.000 & 0.963 & 0.982\\
1.474 & 0.489 & 0.612\\
2.949 & 0.362 & 0.486\\
5.898 & 0.326 & 0.687\\
\bottomrule
\end{tabular}
  \caption{\label{Tab:epsilon_conundrum_coverage} Coverage rates by magnitude of $\beta$ for RL-P using both the lasso and the MCP penalty approximation applied to the Sparse 1 scenario described in Section~\ref{Sec:distribution}. The nominal coverage rate is 80\%. }
\end{table}

The results here provide an example of a version of the RL-P method but with an approximation to the MCP penalty to obtain intervals. The MCP penalty closely resembles that of the lasso near zero but eventually levels out to a constant penalty for larger values of $\beta$ unlike the lasso which applies a penalty proportional to the magnitude of $\beta$. Note that although $\beta$s with intermediate magnitudes are still under covered (albeit with coverage notably higher than RL-P lasso), that the largest $\beta$s have coverage at nearly 70\%, over doubling the coverage of the RL-P lasso.

\clearpage

\section{Additional Details on Bootstrap Bias}\label{Sup:proof}

\begin{figure}[htb!]
  \begin{center}
    \includegraphics[width=0.7\linewidth]{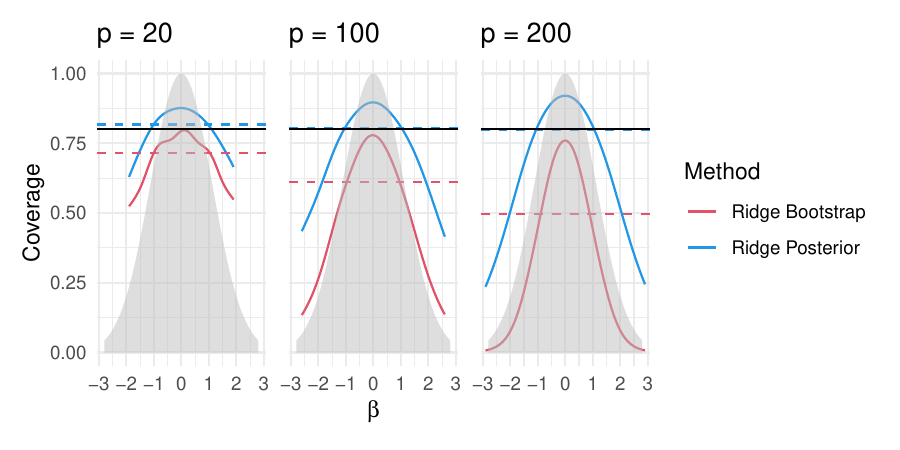}
    \caption{\label{Fig:ridge_converge} Average coverages across p parameters (dashed lines) and estimated coverages as functions of $\beta$ (solid curves) for intervals constructed using a pairs bootstrap (Ridge Bootstrap) and a Bayesian posterior (Ridge Posterior). Full details of the simulation set up can be found in Section~\ref{Sec:boot-bias}.}
  \end{center}
\end{figure}

We continue with a simple proof of the additional bias introduced by bootstrapping in one and two dimensional settings starting first with Ridge regression and then the lasso. After these simple proofs, we provide a simulation that helps generalize this issue to high dimensional settings. Note that here the goal is to provide an intuitive understanding of the issue. Others have already proved this issue more generally \citep{karoui2016, clarté2024}, however, we feel that the lack of straight forward examples may be in part why the bootstrap is still a tool used for penalized regression in general.

Let $\frac{1}{n}x_1^Tx_1 = s_{11}$, and $s_{11}^*$ be the bootstrapped version of $s_{11}$. Furthermore note more that $E_{boot}[s_{11}^*] = s_{11}$. For Ridge, it can be shown that $E(\bh_1) = \frac{s_{11}}{s_{11} + \lam}\beta_1$, letting $g(s_{11}) = \frac{s_{11}}{s_{11} + \lam}$ note that:

$$
\begin{aligned}
g'(s_{11}) &= \lambda (s_{11} + \lam)^{-2} \\
g''(s_{11}) &= -2\lambda (s_{11} + \lam)^{-3} < 0
\end{aligned}
$$

\noindent Thus by Jensen's inequality we have,

$$
E_{boot}[g(s_{11}^*)] \leq g(E_{booot}[s_{11}^*]) = g(s_11).
$$

\noindent Multiplying by $\beta_1$ then gives,

$$
E_{boot}[\bh_1^*] = E_{boot}[g(s_{11}^*)]\beta_1 \leq g(s_{11})\beta_1 = E[\bh_1].
$$

\noindent That is, even in a single parameter setting, just due to the bootstrap variability of $s_{11}$ alone, we would expect bias in the bootstrapped estimate of $\beta_1$.

Now consider a two parameter setting where we assume $\beta_2 = 0$. Let $V = \frac{1}{n}\boldsymbol{X}^T\boldsymbol{X} = \begin{pmatrix} s_{11} & s \\ s & s_{22} \end{pmatrix}$. Similarly here, one can find that

$$
\begin{aligned}
E[\hat{\beta}_1 | V, \lambda] &= \frac{s_{11}(s_{22} + \lambda) - s^2}{(s_{11} + \lambda)(s_{22} + \lambda) - s^2} \beta_1 \\
&= g(s_{11}, s_{22}, s)\beta_1.
\end{aligned}
$$

\noindent Furthermore, with some patience, it is possible to work out the Hessian of $g$:

$$
H = \nabla^{2}g
= \frac{\lambda}{D^{3}}
  \begin{pmatrix}
    -2(s_{22}+\lambda)^{3} & -2(s_{22}+\lambda)s^{2} & 4(s_{22}+\lambda)^{2}s\\[6pt]
    -2(s_{22}+\lambda)s^{2} & -2s^{2}(s_{11}+\lambda) & 2(s_{11}+\lambda)(D+2s^{2})s\\[6pt]
     4(s_{22}+\lambda)^{2}s &  2(s_{11}+\lambda)(D+2s^{2})s & -2(s_{22}+\lambda)(A+3s^{2})
  \end{pmatrix}.
$$

\noindent Because $\lambda/D^{3}>0$, negative semidefiniteness of $H$ is equivalent to that of the scaled matrix $\tilde H = D^{3}H/\lambda$.

$$
\begin{aligned}
\tilde H_{11} &= -2(s_{22}+\lambda)^{3} < 0, \\
\det\bigl(\tilde H_{[1:2,1:2]}\bigr) &= 4(s_{22}+\lambda)^{2}s^{2}D \geq 0, \\
\det(\tilde H) &= -\,8\,\lambda\,(s_{22}+\lambda)^{3}s^{2}\,(s_{11}+\lambda)\,D \leq 0.
\end{aligned}
$$

\noindent As such, $H$ is negative‐semidefinite and consequently $g(s_{11},s_{22},s)$ is globally concave in all three arguments. Now, let $V^*$ be the bootstrapped version of V and note that $E_{boot}[(s_{11}^*, s_{22}^*, s^*)] = (s_{11}, s_{22}, s)$. By Jensen's inequality:

$$
E_{boot}[g(S^*)] \leq g(E_{boot}[S^*]) = g(S).
$$

\noindent Multiplying by $\beta_1$ then gives:

$$
E_{boot}[\hat{\beta}_1^*] = E_{boot}[g(S^*)]\beta_1 < g(S)\beta_1 = E[\hat{\beta}_1].
$$

A similar argument can be made with the lasso, however, the mathematical details become more involved. Whereas with ridge we started off assuming the true values, with lasso it is easier to work with assumed conditions on the lasso estimates themselves as they directly affect the KKT conditions. Starting with a single parameter set up, assume that $\bh_1 > 0$. Then,

$$
\begin{aligned}
&\frac{1}{n}\x_1^T(\y - \x_1 \bh_1) = \lambda \\
&\Rightarrow \bh_1 = \frac{1}{ns_{11}}\x_1^T\y - \frac{\lambda}{s_{11}} \\
&\Rightarrow E[\bh_1] = \beta_1 - \frac{\lambda}{s_{11}} \\
\end{aligned}
$$

\noindent Let $h(s_{11}) = - \frac{\lambda}{s_{11}}$, then

$$
\begin{aligned}
h'(s_{11}) &= \frac{\lambda}{s_{11}^2}, \\
h''(s_{11}) &= \frac{-2\lambda}{s_{11}^3}  < 0,\\
\end{aligned}
$$

\noindent and $h(s_{11})$ is therefore concave. So by Jensen's inequality

$$
\begin{aligned}
E_{boot}[h(s_{11}^*)] \leq h(E_{booot}[s_{11}^*]) = h(s_11).
\end{aligned}
$$

\noindent Adding $\beta_1$ then gives,

$$
\begin{aligned}
E_{boot}[\bh_1^*] = \beta_1 + E_{boot}[h(s_{11}^*)] \leq \beta_1 + g(s_{11}) = E[\bh_1].
\end{aligned}
$$

Now consider a two parameter setting. If we assume $\hat{\beta}_1 > 0$ and $\hat{\beta}_2 = 0$, bias can only be guaranteed when $\beta_2 = 0$. This boils down to the same details as the single parameter case we just considered. If both $\hat{\beta}_1$ and $\hat{\beta}_2$ are assumed positive, we arrive at the following KKT conditions:

$$
\begin{aligned}
\frac{1}{n}\boldsymbol{X}^T\boldsymbol{X} \hat{\boldsymbol{\beta}} &= \frac{1}{n}\boldsymbol{X}^Ty - \lambda 1_2 \\
\end{aligned}
$$

\noindent Again, let $\boldsymbol{V} = \frac{1}{n}\boldsymbol{X}^T\boldsymbol{X} = \begin{pmatrix} s_{11} & s \\ s & s_{22} \end{pmatrix}$, then

$$
\begin{aligned}
\hat{\beta_1} = \frac{s_{22}(x_1^Ty- n\lambda) - s(x_2^Ty- n\lambda)}{n(s_{11}s_{22} - s^2)}
\end{aligned}
$$

\noindent and

$$
\begin{aligned}
E(\hat{\beta}|X) &= \beta_1 - \lambda\frac{s_{22} - s}{s_{11}s_{22} - s^2} \\
&= \beta_1 + h(s_{11}, s_{22}, s).
\end{aligned}
$$

\noindent Differentiation gives
$$
\nabla^{2}h \;=\; \frac{\lambda}{D^{3}}
\begin{pmatrix}
-2s_{22}^{2}(s_{22}-s) & -2s_{22}s^{2} & 4s_{22}^{2}s\\[6pt]
-2s_{22}s^{2} & -2s_{11}s^{2} & 2\,s_{11}s\bigl(D+2s^{2}\bigr)\\[6pt]
4\,s_{22}^{2}s & 2\,s_{11}s\bigl(D+2s^{2}\bigr) & -2s_{22}\bigl(s_{11}s_{22}-s^{2}+3s^{2}\bigr)
\end{pmatrix}.
$$

\noindent With this we find that

$$
\begin{aligned}
H_{11} = -\frac{2\lambda s_{22}^{2}(s_{22}-s)}{D^{3}} < 0, \\ 
\det\bigl(H_{[1:2,1:2]}\bigr) = \frac{4\lambda^{2}s_{22}^{2}s^{2}D}{D^{6}} \geq 0, \\
\det(H) = -\frac{8\lambda^{3}\,s_{22}^{3}s^{2}s_{11}\,s_{11}D}{D^{9}} \leq 0.
\end{aligned}
$$

\noindent Thus, the Hessian is negative-semidefinite.  Therefore the map $g(s_{11},s_{22},s)$ is jointly concave, and multivariate Jensen's inequality implies
$$
E_{\text{boot}}\bigl[\hat\beta_{1}^{*}\bigr]
  \;\le\; E\bigl[\hat\beta_{1}\mid X,\,z_{1}=z_{2}=+1\bigr],
$$
so the bootstrap mean of $\hat\beta_{1}$ sits strictly closer to zero than the original estimate whenever both fitted coefficients are positive. Note that in each of these settings a common thread is that the curvature, which affects the amount of bias introduced by the bootstrap is heavily dependent on $\lambda$. Thus, the larger $\lambda$ is the more the bootstrap will be biased. We also see in the two parameter setting that the correlation between features has a considerable role in the curvature. How much each of these contribute to the bootstrap bias is problem dependent, however, our exploration, such as the example we provide next, suggests that $\lambda$ is likely the larger influencer in general. That said, full exploration is outside the scope of this work, but could be explored by considering the eigenvalues for the hessians above under various scenarios.

How does this extend to high dimensions? To answer this question, we consider the following simulation study.

In order to increase interpretability, we consider a simplified scenario here. Consider a set up with $n = p = 100$ where there is 1 true non-null variable, $A$ s.t. $\beta_A = 2$ that is correlated with a null variable $B$ with $\rho = 0.5$. All other variables are generated independent of $A$ and each other. For this set up, $A$ is always selected to be in the model, both with the original data and for all bootstrap replications at $\lam_{CV}$. This is important as decomposing the bias is more complicated for variables that are not selected to be in the model.

Note, the bias for $\bh_j \neq 0$ is $\frac{1}{n}\x_j^T \epsilon + \frac{1}{n}\x_j^T \X_{-j}(\bb^*_{-j} - \bbh_{-j}) + \lam$. We can break this down further to apply to the scenario outlined above. $Bias_A = \frac{1}{n}\x_A^T \epsilon + \frac{1}{n}\x_A^T \x_{B}(\beta^*_{B} - \bh_{B}) + \frac{1}{n}\x_A^T \X_{N}(\bb^*_{N} - \bbh_{N}) + \lam$. That is, we can decompose the bias into four parts. The first is the irreducible bias which comes from the chance correlation between $\x_A$ and the errors. The last is the bias directly introduced by the lasso penalty. The other two components attribute bias from the single $B$ variable and all 98 $N$ variables respectively. By taking the simulation set up in the previous paragraph and repeating it 1000 times and each time saving the bias attributable to each of the components, we can get an idea of the distribution of the bias components. More specifically, for each generated dataset, we can decompose the bias for the estimates from the original data as well as from the 1000 bootstrap replications. For the bootstrap replications, we can then save the mean bias. Figure~\ref{Fig:bias_decomp_single_B} shows a summary from doing just that. The top panel gives the densities of the mean bootstrap biases across the 1000 repetitions, the middle gives the densities for the biases on the original dataset across the 1000 repetitions, and the bottom give the densities of their paired differences. In this depiction, the contribution of $\lam$ is excluded. Additionally note here that a positive bias is used to indicated bias \emph{towards} zero.

\begin{figure}[hbtp]
    \begin{center}
    \includegraphics[width=\linewidth]{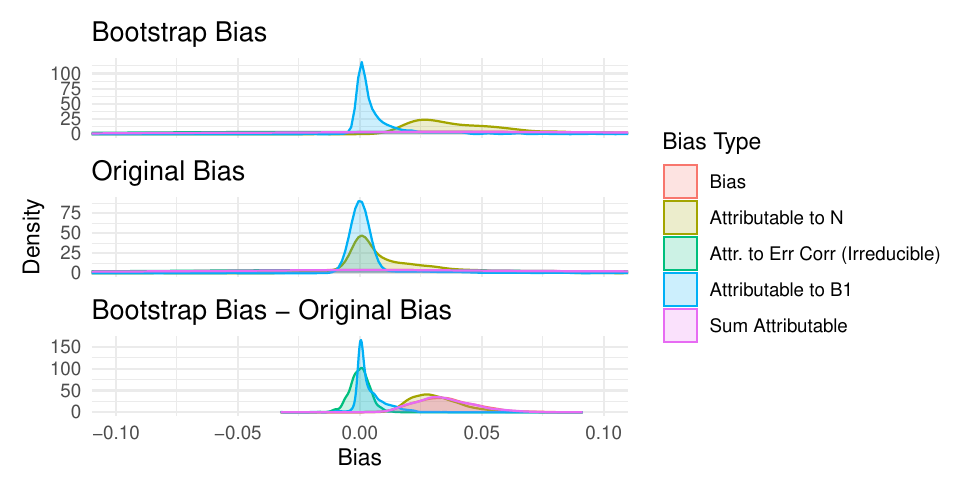}
    \caption{\label{Fig:bias_decomp_single_B} n = p = 100, $\beta_A = 2$, $\beta_B = 0$, $\rho_{A,B} = 0.5$.  All other $\beta$s = 0 and  generated under independence}
    \end{center}
\end{figure}

While the derivation above suggests that increasing correlation is the main contributor to increased bootstrap bias, we see that it is actually the cumulative effect of the 98 N variables that drives the bias rather than the correlation with variable $B$. This has important implications in that even when overall correlations are low, that the sparsity in high dimensional settings is going to lead to a significant bootstrap bias. It is important to differentiate this bias and the bias inherent in the motivation for Section~\ref{Sec:IAC}. The framework put forth in Section~\ref{Sec:IAC} allows for the bias introduced by penalization which we are arguing is permissible in this newly suggested coverage framework, however, the additional bootstrap bias here is what inherently leads to the breakdown of the bootstrap even for average coverage.

\clearpage
\section{A note about stability selection}
\label{sec:stability}

It is no secret that bootstrapping lasso has fundamental issues. The related work in this manuscript is simply meant to show easy to comprehend details to this end. With these issues in mind, it is not common to see a traditional bootstrapping approach applied to the lasso for the purposes of interval construction. Rather, the popular use of resampling techniques for the lasso is in stability selection introduced by  \cite{Meinshausen2010}. However, here, we show that stability selection is still affected by bootstrap bias.

Consider the following set up. $n = 50$, $p = 500$, 4 $\beta$s contain signal: $\beta_{1-4} = (0.25, 0.5, 1, 2)$ and the rest are zero. $\X$ were generated independently from a $\Norm(0, 1)$. Finally, $\y$ was generated as $\y = \X\bb + \bvep$, where $\veps_i \iid N(0, 1)$. Using this data setup, stability selection was performed as outlined in Algorithm~\ref{alg:singlelambda}.

\begin{algorithm}[!ht]
\caption{Bootstrap Stability Selection at a \emph{single} CV‐chosen $\lambda$}
\label{alg:singlelambda}
\begin{algorithmic}[1]
\Require\ $R$ replicated data sets, $B$ bootstraps per data set,
          $p$ predictors
\Statex \textbf{Let:}
  $\mathbf A \in\{0,1\}^{R\times p}$, $\mathbf A^{\!*}\in[0,1]^{R\times p}$
\For{$i=1,\dots,R$}
  \State Generate $(\mathbf X,\mathbf y)$ from the data-generating process
  \State $\lambda_{\text{CV}}\leftarrow
         \operatorname*{arg\,min}_{\lambda}\text{CV‐Error}(\lambda)$
         \Comment{10-fold CV on $(\mathbf X,\mathbf y)$}
  \State Obtain lasso estimates at $\lambda_{\text{CV}}$
        and save
        $\mathbf A_{i,\cdot}\gets\bigl[\hat\beta(\lambda_{\text{CV}})\neq 0\bigr]$

  \State Initialise $\mathbf B^{\!*}\gets\mathbf 0_{B\times p}$
  \For{$b=1,\dots,B$}                        \Comment{pairs bootstraps}
     \State Draw indices $\mathcal I_b$ with replacement from $\{1,\dots,n\}$
     \State $(\mathbf X^{b},\mathbf y^{b})\gets
            (\mathbf X_{\mathcal I_b,\cdot},\mathbf y_{\mathcal I_b})$
     \State Fit lasso on $(\mathbf X^{b},\mathbf y^{b})$ at $\lambda_{\text{CV}}$
     \State $\mathbf B^{\!*}_{b,\cdot}\gets
            \bigl[\hat\beta^{\!*}(\lambda_{\text{CV}})\neq 0\bigr]$
  \EndFor
  \State $\mathbf A^{\!*}_{i,\cdot}\gets
         \dfrac1B\sum_{b=1}^{B}\mathbf B^{\!*}_{b,\cdot}$
\EndFor
\end{algorithmic}
\end{algorithm}

After we computed $\bar{\mathbf A}\gets\dfrac1R\sum_{i=1}^{R}\mathbf A_{i,\cdot}$,  \; $\bar{\mathbf A}^{\!*}\gets\dfrac1R\sum_{i=1}^{R}\mathbf A^{\!*}_{i,\cdot}$ and then compare $\bar{\mathbf A}$ (prob.\ of being selected in the original fit) with $\bar{\mathbf A}^{\!*}$ (expected bootstrap stability) to evaluate how well the inclusion probabilities mirror single-$\lambda$ selection behavior. The results are presented in Table~\ref{Tab:stability_selection}.

\begin{table}[hbtp]
  \centering

\begin{tabular}[t]{ccc}
\toprule
\multicolumn{1}{c}{ } & \multicolumn{2}{c}{Inclusion Probability} \\
\cmidrule(l{3pt}r{3pt}){2-3}
Predictor & Original & Bootstrap\\
\midrule
$\beta_{1}$ & 0.191 & 0.098\\
$\beta_{2}$ & 0.640 & 0.345\\
$\beta_{3}$ & 0.992 & 0.879\\
$\beta_{4}$ & 1.000 & 1.000\\
\bottomrule
\end{tabular}
  \caption{\label{Tab:stability_selection} Results for simulation described in Section~\ref{sec:stability} showing that stability selection also suffers from bootstrap bias. Original selection(\%) provides the empirical selection probabilities for the 4 parameters while Bootstrap stability gives the bootstrap estimate for selection.}
\end{table}

Clearly, there is a fundamental bias for stability selection. That said, further exploration is necessary to understand if this contradicts the claims of \cite{Meinshausen2010} as they primarily focus on FDR control for finite sample problems like the one considered here.

Note, this implementation deviates from the proposed stability selection algorithm originally introduced by \cite{Meinshausen2010} in order to retain a connection to the bootstrap implementation in this manuscript, however, even when considering the entire lasso path fit on sub-bagged samples the bias remained largely the same.

\end{appendices}

\end{document}